\documentclass[notitlepage,superscriptaddress,nofootinbib,twocolumn,tightenlines]{revtex4-2}
\usepackage[utf8]{inputenc}
\usepackage[english]{babel}
\usepackage{booktabs} % Make tables pretty
\usepackage{comment} % Easier way to comment out text
\usepackage{hyperref} %Need for new things, included textorpdfstring
\usepackage{physics} %Load physics commands
\usepackage[table,xcdraw,dvipsnames]{xcolor} % Colored text, can be useful in drafts
\def \sg{\sigma}
\allowdisplaybreaks

\usepackage{amsmath,verbatim,latexsym,amssymb,indentfirst,mathrsfs,mathtools,amsthm,bbm,bm,hyperref,url,cancel,subcaption}
\usepackage[font=small,labelfont=bf,format=plain,justification=raggedright,singlelinecheck=false]{caption}
\usepackage{verbatim,indentfirst}
\usepackage[title,titletoc]{appendix}

\usepackage{graphicx}
\graphicspath{{images/}}
\usepackage{epstopdf}
\newcommand{\caphead}[1]{{\bf #1}}

\renewcommand{\thesection}{\Roman{section}}
\renewcommand{\thesubsection}{\Roman{section} \Alph{subsection}}
\renewcommand{\thesubsubsection}{\Roman{section} \Alph{subsection} \arabic{subsubsection}}
\makeatletter
\def\p@subsection{}
\makeatother
\makeatletter
\def\p@subsubsection{}
\makeatother

\newtheorem{theorem}{Theorem}

\newtheorem{proposition}{Proposition}

\makeatletter
\newcommand\footnoteref[1]{\protected@xdef\@thefnmark{\ref{#1}}\@footnotemark}
\makeatother

\usepackage{soul}

\usepackage{color}
\usepackage[dvipsnames]{xcolor}
\usepackage[normalem]{ulem}

\newcommand{\Group}{\mathcal{G}}
\newcommand{\SUD}{D}   % SU(D)
\newcommand{\su}{\mathfrak{su}}
\newcommand{\Dim}{d}   % Dimensionality of subsystem's (e.g., qubit's) Hilbert space

\newcommand{\Algebra}{\mathcal{A}}

\newcommand{\Ladder}{L}  % Ladder op'rs in sxn about how to construct H^tot
\newcommand{\ad}{\mathrm{ad}}  % Adjoint representation
\newcommand{\DimAlg}{c}  % Dimension of Lie algebra

\newcommand{\Sites}{N}  % # of copies, in the global system, of the system of interest
\newcommand{\hc}{ {\rm h.c.} }

\newcommand{\tot}{ {\rm tot} }
\def\id{\mathbbm{1}}   % Identity
\newcommand{\Hil}{\mathcal{H}}  % Hilbert space
\newcommand{\LParen}{ \bm{(} }
\newcommand{\RParen}{ \bm{)} }
\newcommand{\JParen}{ {(j)} }

\newcommand*{\Set}[1]{\left\{  #1  \right\}}

\renewcommand\th{ {\rm th} }

\newcommand{\shayan}[1]{{\color{black}#1}}

\begin{document}

\title{How to build Hamiltonians that transport noncommuting charges \texorpdfstring{\\}{} in quantum thermodynamics}

\author{Nicole Yunger Halpern}
\email{nicoleyh@umd.edu}
\affiliation{ITAMP, Harvard-Smithsonian Center for Astrophysics, Cambridge, MA 02138, USA}
\affiliation{Department of Physics, Harvard University, Cambridge, MA 02138, USA}
\affiliation{Research Laboratory of Electronics, Massachusetts Institute of Technology, Cambridge, MA 02139, USA}
\affiliation{Center for Theoretical Physics, Massachusetts Institute of Technology, Cambridge, Massachusetts 02139, USA}
\affiliation{Joint Center for Quantum Information and Computer Science, NIST and University of Maryland, College Park, MD 20742, USA}
\affiliation{Institute for Physical Science and Technology, University of Maryland, College Park, MD 20742, USA}
\author{Shayan Majidy} 
\email{smajidy@uwaterloo.ca}
\affiliation{Institute for Quantum Computing, University of Waterloo, Waterloo, Ontario N2L 3G1, Canada}
\affiliation{Perimeter Institute for Theoretical Physics, Waterloo, Ontario N2L 2Y5, Canada} 

\date{\today}

\begin{abstract}
Noncommuting conserved quantities have recently launched a subfield of quantum thermodynamics. In conventional thermodynamics, a system of interest and an environment exchange quantities---energy, particles, electric charge, etc.---that are globally conserved and are represented by Hermitian operators. These operators were implicitly assumed to commute with each other, until a few years ago. Freeing the operators to fail to commute has enabled many theoretical discoveries---about reference frames, entropy production, resource-theory models, etc. Little work has bridged these results from abstract theory to experimental reality. This paper provides a methodology for building this bridge systematically: We present a prescription for constructing Hamiltonians that conserve noncommuting quantities globally while transporting the quantities locally. The Hamiltonians can couple arbitrarily many subsystems together and can be integrable or nonintegrable. Our Hamiltonians may be realized physically with superconducting qudits, with ultracold atoms, and with trapped ions.
\end{abstract}

\maketitle

%\section{Introduction}

One of thermodynamics' most fundamental and ubiquitous interactions is the exchange of quantities between a system of interest and an environment. 
Example quantities include energy, particles, and electric charge.
As the quantities are conserved globally, we call them `charges.' (We call even the local quantities `charges' for convenience, even though the quantities are not conserved locally.) Such exchanges happen, for example, in electrochemical batteries, in a cooling cup of coffee, and when a few spins flip to align with a magnetic field. 
Given such exchanges' pervasiveness, studying their quantum facets is essential for 
(i) developing the field of quantum thermodynamics~\cite{vinjanampathy2016quantum,Goold2015arXiv_review} and 
(ii) discovering nonclassical features of quantum many-body thermalization in condensed matter; atomic, molecular and optical (AMO) physics; high-energy physics; and chemistry. One important quantum phenomenon is operators' failure to commute with each other: Noncommutation underlies uncertainty relations, measurement disturbance, and more. Therefore, studying exchanges of noncommuting charges is crucial for understanding quantum thermodynamics. As a result, noncommuting charges have been enjoying a heyday~\cite{Lostaglio_17_Thermodynamic, Guryanova_16_Thermodynamics, NYH_18_Beyond, Lostaglio_14_Masters, NYH_16_Microcanonical, vaccaro2011information, Sparaciari_18_First, Khanian_20_From, Khanian_20_Resource, Gour_18_Quantum, Manzano_20_Non, Popescu_18_Quantum, Popescu_19_Reference, Ito_18_Optimal, Bera_19_Thermo, Mur_Petit_18_Revealing, Manzano_18_Squeezed,  NYH_20_Noncommuting, Manzano_20_Hybrid, Fukai_20_Noncommutative, Mur-Petit_19_Fluctuations, Scandi_18_Thermodynamic, Manzano_18_Squeezed, Sparaciari_18_First, Mur_Petit_18_Revealing, Boes_18_Statistical, Ito_18_Optimal, Mitsuhashi_21_Characterizing, Croucher2018, wright2018quantum, croucher2021memory} in quantum-information-theoretic (QIT) thermodynamics.

Lifting the assumption that exchanged charges commute~\cite{Jaynes_57_Information_II, Balian_86_Dissipation, Lostaglio_14_Masters,NYH_18_Beyond,Lostaglio_17_Thermodynamic,Guryanova_16_Thermodynamics,NYH_16_Microcanonical, vaccaro2011information} 
has led to discoveries of truly quantum thermodynamics. Example discoveries include 
a generalization of the microcanonical state~\cite{NYH_18_Beyond}, 
resource theories~\cite{Guryanova_16_Thermodynamics,Lostaglio_17_Thermodynamic,YH_16_Microcanonical,NYH_18_Beyond,Sparaciari_18_First,Khanian_20_Resource,Khanian_20_From}, a generalization of the majorization preorder~\cite{Gour_18_Quantum}, 
a reduction of entropy production by charges' noncommutation~\cite{Manzano_20_Non}, 
and reference-frame designs~\cite{Popescu_18_Quantum,Popescu_19_Reference}. 
These discoveries and others have turned noncommuting thermodynamic charges into
a growing subfield.

Most of the discoveries have, until recently, belonged in QIT thermodynamics. However, given their fundamental and nonclassical nature, exchanges of thermodynamic noncommuting charges call for bridges to experiments and to many-body physics. Building these bridges requires Hamiltonians that transport noncommuting observables locally while conserving them globally: As stated in the quantum-thermodynamics review~\citep{vinjanampathy2016quantum}, `an abstract view of dynamics, minimal in the details of Hamiltonians, is often employed in quantum information' and so in QIT thermodynamics. In contrast, experiments, simulations, and many-body theory require microscopic Hamiltonians.

Before the present work, it was unknown 
(i) whether Hamiltonians that transport noncommuting observables locally, while conserving them globally, exist; 
(ii) how such Hamiltonians look, if they exist; 
(iii) how to construct such Hamiltonians for given noncommuting charges; and
(iv) for which charges such Hamiltonians can be constructed. 
We answer these questions, enabling the system-and-environment exchange of noncommuting charges to progress from its QIT-thermodynamic birthplace to many-body physics and experiments. Example predictions that merit experimental exploration include 
(i) the emergence of the quantum equilibrium state in \citep{Lostaglio_17_Thermodynamic, Guryanova_16_Thermodynamics, NYH_18_Beyond}, 
(ii) the decrease in entropy production by noncommuting charges \citep{Manzano_20_Non}, (iii) applications of the entropy decrease to quantum engines \citep{quan2007quantum}, 
(iv) the conjecture that noncommuting charges hinder thermalization \citep{NYH_16_Microcanonical}, and 
(v) the conjecture's application to quantum memories. 
We open the door to experiments by prescribing 
how to construct the needed Hamiltonians.
Our construction also enables the generalization, to noncommuting charges,
of many-body--thermalization tools in condensed matter, AMO physics, and high-energy. Examples include the eigenstate thermalization hypothesis, out-of-time-ordered correlators, and random unitary circuits (e.g.,~\cite{Deutsch_91_Quantum,Srednicki_94_Chaos,Rigol_08_Thermalization,D'Alessio_16_From,Brown_12_Scrambling,Nahum_18_Operator,Khemani_18_Operator,HunterJones_18_Operator,Swingle_18_Quantitative}).

This paper introduces a prescription for constructing Hamiltonians that overtly move noncommuting charges between subsystems while conserving the charges globally. The charges form a finite-dimensional semisimple complex Lie algebra. The Hamiltonians can couple arbitrarily many subsystems together and can be integrable or nonintegrable. The prescription also produces a convenient basis for the algebra---a basis of charges explicitly transported locally, and conserved globally, by the Hamiltonian. 
The prescription is general, being independent of any physical platforms. Consequently, the Hamiltonians can be realized with diverse physical systems, such as superconducting circuits, ultracold atoms, and trapped ions.

In a special case, the charges form the Lie algebra $\su(\SUD)$, $\Sites$ identical subsystems form the global system, and each subsystem corresponds to the Hilbert space $\mathbb{C}^\SUD$. In this example the Schur-Weyl duality describes the Hamiltonians' forms \cite{goodman_2009_symmetry, das_2014_lie}: 
Let the global system (formed from the system of interest and the environment)
be many copies of the system of interest.
The Hamiltonians are the linear combinations of the permutations of the copies. (Hamiltonians have also been engineered to have SU$(\SUD)$ symmetry
without regard to whether noncommuting charges are transported~\cite{Choi_17_Dynamical,Choi_20_Robust}.) Our results are more general than the Schur-Weyl duality and elucidate the dynamics' physical interpretation. First, our prescription governs a much wider class of algebras: all finite-dimensional, semisimple Lie algebras in which the Killing form induces a metric. Many physically significant algebras satisfy these assumptions---for example, the simple Lie algebras, which include $\mathfrak{su}(D)$. Second, our results are not restricted to systems  whose Hilbert spaces are $\mathbb{C}^{D}$. Finally, the Hamiltonian form specified by the Schur-Weyl duality---a linear combination of permutations---is an abstract construct. How to implement an arbitrary linear combination of permutations is not obvious. In contrast, our Hamiltonians have a clear physical interpretation, manifestly transporting noncommuting charges between subsystems. To our knowledge, no other class of Hamiltonians that transport charges locally and conserve them globally, comparably general to our class, is known.

This paper begins with our setup, detailed in Sec.~\ref{sec_Setup}. Section~\ref{sec_Gentle_Intro} introduces the Hamiltonian-construction prescription pedagogically. We also review mathematical background and illustrate the prescription with an example familiar in quantum information, the Lie algebra $\su(2)$. Section~\ref{sec_Algorithm} synopsizes the prescription, crystallizing the main result, and presents two properties of the prescription. A richer example provides intuition in Sec.~\ref{sec_su3}: Hamiltonians that transport and conserve charges in the Lie algebra $\mathfrak{su}(3)$. Section~\ref{sec_Outlook} concludes with potential realizations of our Hamiltonians in condensed matter, AMO, and high-energy and nuclear physics.

%\section{Results}
\section{Setup}\label{sec_Setup}

Consider a global closed quantum many-body system, 
as in recent thermalization experiments~\cite{ritter_07_observing, trotzky_12_probing, langen_13_local, kaufman_16_quantum, smith_16_many, neill_16_ergodic, tang_18_thermalization, malvania_18_onset, sanchez_19_emergent, landsman_19_verified, lewis_19_unifying,  joshi_20_quantum}. As in conventional statistical mechanics, the global system is an ensemble of $\Sites$ identical subsystems.
(We use the term `ensemble' in the traditional sense of statistical physics: a collection of many identical copies of a system of interest. Such ensembles are often invoked to determine equilibrium probability distributions~\cite[p.~62]{mandl1971statistical}.) A few of the subsystems form the system of interest; and the rest, an effective environment. Each subsystem corresponds to a Hilbert space $\Hil$ of finite dimensionality $\Dim$.

We will construct global Hamiltonians, $H^\tot$, that conserve extensive charges defined as follows. Let $Q_\alpha$ denote a Hermitian operator defined on $\Hil$. We denote by $Q_\alpha^\JParen$ the observable defined on the $j^\th$ subsystem's $\Hil$. Each global observable 
\begin{align}
   Q_\alpha^\tot  
   % % %
   :=  \sum_{j = 1}^{\Sites}  Q_\alpha^\JParen 
   % % %
   \equiv   \sum_{j = 1}^\Sites
   \id^{\otimes (j - 1)}  \otimes
   Q_\alpha^\JParen  \otimes
   \id^{\otimes (\Sites - j) }
   % % %
%   & \equiv 
%   \left(  Q_\alpha^\1 \otimes  \id^{\otimes (\Sites - 1)}  \right) 
%   +  \left(  \id  \otimes  Q_\alpha^\2  \otimes  \id^{\otimes (\Sites - 2)}  \right) 
%   \nonumber \\ & \quad
%   + \ldots +
%   \left(   \id^{\otimes (\Sites - 1)}  \otimes  Q_\alpha^{(\Sites)}   \right) 
\end{align}
will be conserved by design:
\begin{align}
   \label{eq_Conserve}
   [H^\tot,  Q_\alpha^\tot]  =  0 .
\end{align}
Although the local $Q_\alpha^\JParen$ are not conserved, we will sometimes call them, and the $Q_\alpha$, `charges' for convenience. One might know, initially, of only $c'$ charges' existence.

These $c'$ $Q_\alpha$'s generate a complex Lie algebra $\Algebra$,
which we assume to be finite-dimensional.
$\Algebra$ consists of all the charges 
(as well as non-Hermitian operators, which we ignore).
Lie algebras describe many conserved physical quantities: 
particle number, angular momentum, electric charge, color charge, 
weak isospin, and our space-time’s metric~\cite{das_2014_lie,iachello_06_lie,gilmore_12_lie}.
We focus on non-Abelian Lie algebras, 
motivated by quantum thermodynamics that highlights noncommutation:
The commutator exemplifies the Lie bracket,
$[Q_\alpha, Q_\beta]$.
% the commutator is one the most commonly studied Lie brackets.

We assume four more properties of the algebra, to facilitate our proofs.
$\Algebra$ is finite-dimensional and semisimple. 
Representing an observable, $\Algebra$ is over the complex numbers.
Also, on $\Algebra$ is defined a Killing form (reviewed below)
that induces a metric. Many physically significant algebras satisfy these assumptions---for
example, the simple Lie algebras
(see the Supplementary Note 1 and~\cite{das_2014_lie,iachello_06_lie,gilmore_12_lie}).

\section{Pedagogical explanation}
\label{sec_Gentle_Intro}

This section describes the prescription for constructing Hamiltonians $H^\tot$
that conserve noncommuting charges globally [Eq.~\eqref{eq_Conserve}]
while transporting them locally:
\begin{align}
   \label{eq_Local_Transport}
   \left[  H^\tot,  Q_\alpha^\JParen  \right]  \neq  0 
\end{align}
for some site $j$.
(In every such commutator throughout this paper, 
one argument implicitly contains tensor factors of $\id$,
so that both arguments operate on the same Hilbert space.) 
We construct two-body interaction terms, 
then combine them into many-body terms.
This explanation provides a pedagogical introduction;
the prescription is synopsized in Sec.~\ref{sec_Algorithm}.
Here, we illustrate each step with an algebra familiar in quantum information, $\su(2)$, which describes spin-$1/2$ angular momentum.

Table~1 lists the simple Lie algebras.
Every Cartesian product of simple Lie algebras yields
a semisimple Lie algebra $\Algebra$.
% https://en.wikipedia.org/wiki/Semisimple_algebra
Such an algebra generates a semisimple Lie group $\Group$.
For example, if $\Algebra$ consists of angular momentum, 
$\Algebra = \su(\SUD)$.
The corresponding $\Group$ consists of rotations:
$\Group = \text{SU}(\SUD)$.

An algebra has two relevant properties, a dimension and a rank
(Table~1).
The dimension, $\DimAlg$, equals
the number of generators in a basis for the algebra.(We chose the notation $\DimAlg$ to evoke
the $c$ introduced in~\cite{NYH_16_Microcanonical}.
There, $c$ was defined as the number of charges.
As explained in the present paper's Sec.~\ref{sec_Setup},
those charges would form a Lie algebra.
Infinitely many charges would therefore exist,
the $c$ in~\cite{NYH_16_Microcanonical} would equal infinity,
and results in~\cite{NYH_16_Microcanonical} would be impractical.
We therefore define $c$ as the Lie algebra's finite dimension.) For example, $\mathfrak{su}$(2) has the Pauli-operator basis
$\Set{  \sigma_x,  \sigma_y,  \sigma_z  }$ 
and so has a dimension $\DimAlg = 3$. 
The rank, $r$, has a significance that we will encounter shortly.

%
%
% Table: Dimension, rank, and D/r for classical simple Lie groups
%
% Reference: "Meeting notes" folder --> "Mark Mueller - 7/16/19"
\begin{table}[t]
\begin{tabular}{@{}cccc@{}} 
   \toprule
        Algebra
   &   Dimension ($\DimAlg$)
   &   Rank ($r$)
   &   $\DimAlg / r$  \\ \midrule
         $\mathfrak{so}$(2$\SUD$)
   &   $\SUD (2 \SUD - 1)$
   &   $\SUD$
   &    $2 \SUD - 1$  \\
        $\mathfrak{sl}$($\SUD$ + 1)
   &   $(\SUD + 1)^2 - 1$
   &   $\SUD$   
   &   $\SUD + 2$    \\
         $\mathfrak{so}$($2\SUD + 1$)
   &   $\SUD (2 \SUD + 1)$
   &   $\SUD$
   &    $2 \SUD + 1$    \\
        $\mathfrak{sp}$($2 \SUD$)
        % Reference: https://en.wikipedia.org/wiki/Symplectic_group
   &   $\SUD (2 \SUD + 1)$ 
   &   $\SUD$
   &    $2 \SUD + 1$  \\
    $\mathfrak{g}_2$ & $14$  & $2$ & $7$ \\
    $\mathfrak{f}_4$ & $52$ & $4$ & $13$ \\
    $\mathfrak{e}_6$ & $78$ & $6$ & $13$ \\
    $\mathfrak{e}_7$ & $133$ & $7$ & $19$ \\
    $\mathfrak{e}_8$ & $248$ & $8$ & $31$ \\
   \bottomrule
\end{tabular}
\caption{\caphead{Simple Lie algebras:}
$\DimAlg$ denotes an algebra's dimension, and $r$ denotes the rank. 
We implicitly omit $\mathfrak{so}(2)$ and $\mathfrak{so}(4)$,
which are not simple~\cite{gilmore_12_lie}. Also $\su(\SUD)$ is a simple Lie algebra. However, including $\su(\SUD)$ would be redundant: The complexification of $\su(\SUD)$ is isomorphic to $\mathfrak{sl}(\SUD)$.}
\label{table_D_and_r}
\end{table}

% With every simple Lie algebra can be associated a Killing form.
% En route to defining the Killing form,
% we review the adjoint representation and general forms.
A representation of $\Algebra$ is a Lie-bracket-preserving map
from $\Algebra$ to a set of linear transformations.
The adjoint representation maps from $\Algebra$ to 
linear transformations defined on $\Algebra$.
If $x \in \Algebra$, the adjoint representation $\ad(x)$ 
acts on $y \in \Algebra$ as
$\ad(x) (y)  :=  [x, y]$.
The adjoint representation features in the Killing form,
which we review now.
The definition of $\Algebra$ involves a vector space $V$
defined over a field $F$.
A map $V \times V \to F$ is a \emph{form}.
The \emph{Killing form} is the symmetric bilinear form
\begin{align}
   (x, y)  :=  \Tr \LParen  \ad(x)  \ad(y)  \RParen .
\end{align}
% for $x, y \in \Algebra$.
We say that $x$ and $y$ are \emph{Killing-orthogonal} if
$(x, y) = 0$.
We say that subalgebras $\Algebra_1$ and $\Algebra_2$
are Killing-orthogonal if, for all $x \in \Algebra_1$ and $y \in \Algebra_2$,
$(x, y) = 0$.
We will use the Killing form to construct
the preferred basis of charges for $\Algebra$.

Our construction begins with another basis:
Every finite-dimensional semisimple complex Lie algebra $\Algebra$ 
has a Cartan-Weyl basis.
% Reference: email chain "Conditions under which results hold" -- first message sent by Shayan on 3/11/21: "All finite-dimensional semisimple complex Lie algebras have a Cartan-Weyl basis. I know that the semisimple condition is necessary and that three conditions together are sufficient."
In fact, $\Algebra$ has infinitely many.
Convention may distinguish one Cartan-Weyl basis.
We use the conventional $\su(2)$ basis for concreteness.
We use this basis, in our example, for concreteness.
In general, one selects an arbitrary Cartan-Weyl basis.
The basis contains generators of two types:
Hermitian operators and ladder operators.

The number of Hermitian operators is the algebra's rank, $r$.
These operators commute with each other.
If $r > 1$, we rescale the operators to endow them with unit Hilbert-Schmidt norms:
\begin{align}
   \label{eq_Q_Norms}
   \Tr (Q_\alpha^\dag Q_\alpha) = 1 .
\end{align}
We include these operators, $Q_{\alpha = 1, 2, \ldots, r}$,
in our preferred basis.
In the $\mathfrak{su}$(2) example, $r = 1$; and 
$Q_1 = \sigma_z$, whose eigenstates $\ket{ \pm z }$
correspond to the eigenvalues $\pm 1$.
The $Q_\alpha$'s generate a subalgebra, a \emph{Cartan subalgebra}.

The Cartan-Weyl basis contains, as well as Hermitian operators, ladder operators.
They form pairs $L_{\pm \beta}$, for
$\beta = 1, 2, \ldots, \frac{\DimAlg - r}{2}$: Since the Cartan-Weyl basis has $c$ elements, and $r$ of them are Hermitian, there are $c-r$ ladder operators. Each $\beta$ corresponds to two ladder operators, one raising $(+\beta)$ and one lowering $(-\beta)$. Hence $\beta$ runs from $1$ to $\frac{\DimAlg - r}{2}$. Each $L_{\pm \beta}$ raises or lowers at least one $Q_\alpha$.
In the $\mathfrak{su}$(2) example, the ladder operators
$\sigma_{\pm z}  =  \frac{1}{2} ( \sigma_x  \pm  i  \sigma_y )$
raise and lower $\sigma_z$:
$L_{\pm z}  \ket{ \mp z }
=  \ket{ \pm z }$.
In other algebras, an $L_{\pm \beta}$ 
can raise and/or lower multiple $Q_\alpha$'s.
Examples include $\su(3)$ (Sec.~\ref{sec_su3}).

From each ladder-operator pair, we construct 
an interaction that couples subsystems $j$ and $j'$.
Let $J_\beta^{(j, j')}$ denote a hopping frequency.
An interaction that transports all the charges between $j$ and $j'$,
while conserving each charge globally, has the form
\begin{align}
   \label{eq_H_2body_1}
   H^{(j, j')}
   % % %
   \propto  \sum_{\beta = 1}^{ (c - r) / 2 }  J_\beta^{(j, j')}    
   \left( L_{+\beta}^\JParen  L_{-\beta}^{(j')}  
           +  L_{-\beta}^\JParen  L_{+\beta}^{(j')}  \right) .
\end{align}
% Over which $j'$ values the sum runs may depend on experimental capabilities.

We assemble the other terms in $H^{(j, j')}$ from other Cartan-Weyl bases,
constructed as follows.
Let $U$ denote a general element of the group $\Group$.
We conjugate, with $U$, each element of our first Cartan-Weyl basis: 
For $\alpha = 1, 2, \ldots, r$ and 
$\beta = 1, 2, \ldots, \frac{c - r}{2}$,
\begin{align}
   \label{eq_New_CW_Basis}
   & Q_\alpha  \mapsto  U^\dag Q_\alpha U
   = Q_{\alpha + r},
   % % %
   \quad \text{and}  \\
   % % %
   & L_{\pm \beta}  \mapsto
   U^\dag  L_{\pm \beta} U 
   = L_{\pm \left( \beta + \frac{c - r}{2} \right) }.
\end{align}
We include the new $Q_{\alpha}$'s
(for which $\alpha = r + 1, r + 2, \ldots, 2r$) in
our preferred basis for the algebra.

We constrain $U$ such that each new $Q_{\alpha}$ 
is Killing-orthogonal to
(i) each other new charge $Q_{\beta}$ and 
(ii) each original charge $Q_{\gamma}$:
\begin{align}
   \label{eq_Orth_Condn}
   (Q_{\alpha}, Q_{\beta})  
   =  (Q_{\alpha},  Q_{\gamma})
   =  0
\end{align}
for all $\alpha, \beta  = r + 1, r + 2,  \ldots,  2r$
and all $\gamma = 1, 2, \ldots, r$.
This orthogonality restricts $U$, though not completely.
The new $Q_{\alpha}$'s generate a Cartan subalgebra
Killing-orthogonal to the original Cartan subalgebra.
The new ladder operators contribute to the interaction:
\begin{align}
   \label{eq_H_2body_2}
   H^{(j, j')}  
   % % %
   \propto  \sum_{\beta = 1}^{c - r}
   J_{\beta}^{(j, j')}
   \left(  L_{+\beta}^\JParen    L_{-\beta}^{(j')}
            +  \hc  \right) .
\end{align}

In the $\mathfrak{su}$(2) example, $U$ can be represented by
% Reference: https://en.wikiversity.org/wiki/SU(2)#SU(2)
$\begin{bmatrix}
      a  &  -b^* \\
      b  &  a^*
\end{bmatrix} \, ,$
wherein $a, b \in \mathbb{C}$ and $| a |^2  +  | b |^2  =  1$.
The prescription restricts $U$ only via
the Killing-orthogonality of $U^\dag \sigma_z U$ to $U$.
We enforce only this restriction in the Supplementary Note 2.
Here, we choose a $U$ for pedagogical simplicitly:
$U = ( \id + i\sigma_y ) / \sqrt{2} $, such that 
$Q_{\alpha + r} = Q_2 
=  \sigma_x$.
The new ladder operators,
$\sigma_{\pm x}
:=  \left(  \frac{ \id + i\sigma_y }{ \sqrt{2} }  \right)
\sigma_{\pm z}
\left( \frac{ \id + i\sigma_y }{ \sqrt{2} }   \right)$,
create and annihilate quanta of 
the $x$-component of the angular momentum.
The interaction becomes
\begin{align}
   H^{(j, j')}  
   % % %
   \propto  \sum_{\beta = z, x}  
   J_\beta^{(j, j')}
   \left(  \sigma_{+ \beta}^\JParen
   \sigma_{- \beta}^{(j')}  +  \hc  \right) .
\end{align}

We repeat the foregoing steps: Write out the form of a general $U \in \Group$.
Conjugate each element of the original Cartan-Weyl basis with $U$.
Constrain $U$ such that
the new $Q_\alpha$'s are orthogonal to each other
and to the older $Q_\alpha$'s.
Include the new $Q_\alpha$'s in our preferred basis for the algebra.
Form a term, in $H^{(j, j')}$, from the new ladder operators $L_{\pm \beta}$.

Each Cartan-Weyl basis contributes $r$ elements $Q_\alpha$ to
the preferred basis.
The basis contains $\DimAlg$ elements,
so we form $\DimAlg / r$ mutually orthogonal Cartan-Weyl bases.
$\DimAlg / r$ equals an integer for the finite-dimensional semisimple complex Lie algebras, 
according to Proposition~\ref{prop_Integer_Ratio} in Sec.~\ref{sec_Algorithm}.
Table~1 confirms the claim for the simple Lie algebras.
Our algebra's finite dimensionality ensures that our prescription halts.
The two-body interaction is now
\begin{align}
   \label{eq_H_2body_3}
   H^{(j, j')}  
   % % %
   =  \sum_{\beta = 1}^{ \frac{c - r}{2}  \cdot  \frac{c}{r} }
   J_{\beta}^{(j, j')}
   \left(  L_{+\beta}^\JParen    L_{-\beta}^{(j')}
            +  \hc  \right) .
\end{align}

Why is the preferred basis $\{ Q_\alpha \}$ preferable?
First, the basis endows the Hamiltonian with a simple physical interpretation:
$H^{(j, j')}$ transports all these charges locally while conserving them globally.
Second, the basis is (Killing-)orthogonal.

In the $\mathfrak{su}$(2) example, $\DimAlg / r = 3 / 1 = 3$.
Hence we construct three Cartan-Weyl bases, using two SU(2) elements. 
If the first unitary was  $( \id + i\sigma_y ) / \sqrt{2}$,
the second unitary is
$( \id - i \sigma_x + i  \sigma_y + i \sigma_z ) / 2$, 
to within a global phase. 
Consequently, $Q_3 = \sigma_y$, the preferred basis for $\Algebra$ is
$\{ \sg_z, \sg_x, \sg_y \}$, and
\begin{align}
   \label{eq_SU2_Ham_1}
   H^{(j, j')}  
   % % %
   =  \sum_{\beta = x, y, z}  
   J_\beta^{(j, j')}
   \left( \sigma_{+ \beta}^\JParen  \sigma_{- \beta}^{(j')}
     +  \hc  \right) .
\end{align}

Next, we constrain the interaction to conserve every global charge:
\begin{align}
   \label{eq_Comm_Condn}
   [ H^{(j, j')},  Q_\alpha^\tot ]  =  0
   \quad \forall \alpha = 1, 2, \ldots, c .
\end{align}
The commutation relations~\eqref{eq_Comm_Condn} constrain 
the hopping frequencies $J_\alpha^{(j, j')}$.
The frequencies must equal each other in the $\su(2)$ example:
$J_\alpha^{(j, j')}  \equiv  J^{(j, j')}$ for all $\alpha$. 
The Hamiltonian simplifies to~\cite{NYH_20_Noncommuting}
\begin{align}
   \label{eq_Heisenberg}
   H^{(j, j')}  =  J^{(j, j')}
   \vec{\sigma}^\JParen  \cdot  \vec{\sigma}^{(j')} .
\end{align} 
This Heisenberg model is known to have SU(2) symmetry 
and so to conserve each global spin component 
$\sigma_\alpha^\tot :=\sum_{j=1}^{N}\sigma_{\alpha}^{(j)}$. But the Hamiltonian is typically written in 
the dot-product form~\eqref{eq_Heisenberg}, as
\begin{align}
   \label{eq_su2_H_Simple}
   H^{(j, j')}  
   \propto  \sum_{\alpha = x, y, z}
   \sigma_\alpha^\JParen  \sigma_\alpha^{(j')} .
\end{align}
or in the $z$-biased form 
$H^{(j, j')}  \propto
2 (\sigma_{+z}^\JParen  \sigma_{-z}^{(j')}
+  \sigma_{-z}^\JParen  \sigma_{+z}^{(j')} )
+  \sigma_z^\JParen  \sigma_z^{(j')}$.
None of these three forms reveals that the Heisenberg model
transports noncommuting charges between subsystems.
Our expression~\eqref{eq_SU2_Ham_1} and our prescription do.
In relativistic field theories, making the action 
manifestly Lorentz-invariant is worthwhile; 
analogously, making the Hamiltonian manifestly 
transport noncommuting charges locally, 
while conserving them globally, is worthwhile.
Furthermore, our prescription constructs Hamiltonians that
overtly transport noncommuting charges locally and conserve the charges globally
not only in this simple $\su(2)$ example,
but also for all finite-dimensional semisimple complex Lie algebras
on which the Killing form induces a metric---including algebras for which
this prescription does not produce the Heisenberg Hamiltonian. Supplementary Note 3 discusses
a generalization of the simple form~\eqref{eq_Heisenberg}.

We have constructed a two-body interaction $H^{(j, j')}$
that couples subsystems $j$ and $j'$.
We construct $k$-body terms $H^{ \left( j,  j',  \ldots, j^{(k)}  \right) }$ 
by multiplying two-body terms~\eqref{eq_H_2body_3} together,
constraining the couplings such that 
$[ H^{ \left( j,  j',  \ldots, j^{(k)}  \right) }, Q_\alpha^\tot] = 0$,
and subtracting off any fewer-body terms that appear in the product.
Section~\ref{sec_Algorithm} details the formalism.
In the $\su(2)$ example, a three-body interaction has the form 
(see Supplementary Note 2)
\begin{align} 
   H^{(j, j', j'')}  
   % % %
   & \propto  H^{(j, j')}  H^{(j', j'')}  H^{(j'', j)} \\
   % % %
   \label{eq_3Body_H_su2}
   & \propto J^{(j,j',j'')}   [(\sg_x\sg_y\sg_z + \sg_y\sg_z\sg_x + \sg_z\sg_x\sg_y) 
   \nonumber \\ & \quad 
    - ( \sg_z\sg_y\sg_x + \sg_x\sg_z\sg_y + \sg_y\sg_x\sg_z) ] .
\end{align}
wherein $J^{(j,j',j'')}  \in  \mathbb{R}$.

The Hamiltonian we constructed may be integrable.
For example, the one-dimensional (1D) nearest-neighbor Heisenberg model
is integrable~\cite{baxter_16_exactly}.
Integrable Hamiltonians have featured in studies of noncommuting charges in thermodynamics~\cite{Fukai_20_Noncommutative}.
But one might wish for the system to thermalize as much as possible, 
as is promoted by nonintegrability~\citep{gogolin_16_equilibration,D'Alessio_16_From}.
Geometrically nonlocal couplings, many-body interactions, 
and multidimensional lattices tend to break integrability.
Hence one can add terms $H^{(j, j')}$ and $H^{(j, j', \ldots, j^{(k)} )}$
to the global Hamiltonian $H^\tot$,
and keep growing the lattice's dimensionality,
until $H^\tot$ becomes nonintegrable.
Nonintegrability may be diagnosed with, e.g., energy-gap statistics~\cite{D'Alessio_16_From}.
In the $\su(2)$ example, one can break integrability by 
creating next-nearest-neighbor couplings
or by making the global system two-dimensional~\cite{NYH_20_Noncommuting}.

\section{Prescription for constructing the Hamiltonians}
\label{sec_Algorithm}

Here, we synopsize the prescription elaborated on in Sec.~\ref{sec_Gentle_Intro}.
Then, we present two results pertinent to the prescription.
% (i) The $Q_\alpha$ produced by the prescription}
% form a basis for the algebra $\Algebra$.
% (ii) Every finite-dimensional semisimple complex Lie algebra has a dimension $c$ and a rank $r$ that form an integer ratio $c/r$, as required by the prescription}.
% (iii) The Hamiltonian may assume a simple form in some cases.
We construct, as follows, Hamiltonians that transport noncommuting charges locally
and conserve the charges globally:
\begin{enumerate}% [label=(\roman*)]

   \item 
   Identify an arbitrary Cartan-Weyl basis for the algebra, $\Algebra$. 
   
   \item
   The Cartan-Weyl basis contains $r$ Hermitian operators that commute with each other. 
   Scale each such operator such that it has a unit Hilbert-Schmidt norm 
   [Eq.~\eqref{eq_Q_Norms}].
   Label the results $Q_{\alpha = 1, 2, \ldots, r}$.
   Include them in the preferred basis for the algebra.
   
   \item
   The other Cartan-Weyl-basis elements are ladder operators 
   that form raising-and-lowering pairs:
   $\Ladder_{\pm \beta}$, for
   $\beta = 1, 2, \ldots \DimAlg - r$.
   From each pair, form one term in the two-body interaction, $H^{(j, j')}$
   [Eq.~\eqref{eq_H_2body_1}].
   
   \item \label{step_U} 
   Write out the form of the most general element 
   $U  \in  \Group$ of the Lie group $\Group$ generated by $\Algebra$. 
   Conjugate each charge $Q_\alpha$ and each ladder operator $L_{\pm \beta}$
   with $U$ [Eq.~\eqref{eq_New_CW_Basis}].
   The new charges and new ladder operators, together, 
   form another Cartan-Weyl basis.
   
   \item 
   Constrain $U$ such that every new charge $Q_{\alpha}$ is Killing-orthogonal to
   (i) each other new charge and
   (ii) each charge already in the basis [Eq.~\eqref{eq_Orth_Condn}].
   
   \item
   Include each new $Q_{\alpha}$ in the basis for $\Algebra$.
   
   \item \label{step_Add_To_H2}
   From each new pair $L_{\pm \beta}$ of ladder operators,
   form a term in the two-body interaction $H^{(j, j')}$ 
   [Eq.~\eqref{eq_H_2body_2}].
   
   \item \label{step_Repeat}
   Repeat steps~\ref{step_U}-\ref{step_Add_To_H2} until having identified
   $\DimAlg / r$ Cartan-Weyl bases,
   wherein $\DimAlg$ denotes the algebra's dimension.
   Each Cartan-Weyl basis contributes $r$ elements $Q_\alpha$ to
   the preferred basis for $\Algebra$.
   The basis is complete, containing
   $r \cdot \frac{\DimAlg}{r} = \DimAlg$ elements.
   
   \item 
   Constrain the two-body interaction to conserve each global charge
   [Eq.~\eqref{eq_Comm_Condn}], for all $\alpha = 1, 2, \ldots, \DimAlg$.
   Solve for the frequencies $J_\beta^{(j, j')}$ that satisfy this constraint.
   
   \item
   If a $k$-body interaction is desired, for any $k > 2$:
   Perform the following substeps for $\ell = 3, 4, \ldots, k$:
   Multiply together $\ell$ unconstrained two-body interactions~\eqref{eq_H_2body_3} cyclically:
   % What if the multiplication isn't cyclic? See `Ladder Hamiltonian' folder --> "Notes" --> "k-body H - noncyclic multiplicn"
   \begin{align}
      \label{eq_k_body_H}
      H^{ \left( j,  j',  \ldots, j^{(\ell)}  \right) }
      % % %
      &= H^{(j, j')}  H^{(j', j'')}  \ldots 
      H^{ \left( j^{(\ell-1)},  j^{(\ell)} \right) }  
      \nonumber \\ & \quad \times
      H^{( j^{(\ell)},  j )} .
   \end{align}
   Constrain the couplings so that 
   $[ H^{ \left( j,  j',  \ldots, j^{(\ell)}  \right) },  Q_\alpha^\tot ]  =  0$
   for all $\alpha$.
   If $H^{ \left( j,  j',  \ldots, j^{(\ell)}  \right) }$ contains fewer-body terms
   that conserve all the $Q_\alpha^\tot$, subtract those terms off.
   
   \item
   Sum the accumulated interactions $H^{ \left(j, j', \ldots, j^{(k)} \right)}$ 
   over the subsystems $j, j', \ldots$ to form $H^\tot$.
   
   \item
   If $H^\tot$ is to be nonintegrable, add longer-range interactions and/or
   large-$k$ $k$-body interactions until breaking integrability,
   as signaled by, e.g., energy-gap statistics.

\end{enumerate}

Having synopsized our prescription, we present two properties of it.
The first property ensures that the prescription runs for
an integer number of iterations (step~\ref{step_Repeat}).

% Proposition
\begin{proposition}
\label{prop_Integer_Ratio}
Consider any finite-dimensional semisimple complex Lie algebra.
The algebra's dimension, $c$, and rank, $r$, form an integer ratio:
$c / r \in \mathbb{Z}_{>0}$.
\end{proposition}
\noindent
We prove this proposition in the Supplementary Note 4.
The second property characterizes the prescription's output.

% Theorem
\begin{theorem}
\label{thm_Qs_Form_Basis}
The charges $Q_1, Q_2,  \ldots,  Q_c$ produced by the prescription 
form a basis for the algebra $\Algebra$.
\end{theorem}

\begin{proof}
The charges are Killing-orthogonal by construction:
$(Q_\alpha,  Q_\beta) = 0$ for all $\alpha, \beta$. 
The Killing form induces a metric on $\Algebra$ by assumption.
Therefore, the $Q_\alpha$ are linearly independent
according to this metric.

The prescription produces $c$ charges (step~\ref{step_Repeat}). 
$c$ denotes the algebra's dimension, the number of elements in
each basis for $\Algebra$.
Hence every linearly independent set of $c$ $\Algebra$ elements 
forms a basis for $\Algebra$.
Hence the $Q_\alpha$ form a basis.
\end{proof}

\section{\texorpdfstring{$\su(3)$}{su(3)} example}
\label{sec_su3}

Section~\ref{sec_Gentle_Intro} illustrated the Hamiltonian-construction prescription
with the algebra $\su(2)$. 
The $\su(2)$ example offered simplicity but lacks other algebras' richness:
In other algebras, each Cartan-Weyl basis contains
multiple Hermitian operators and multiple ladder-operator pairs.
We demonstrate how our prescription accommodates this richness,
by constructing a two-body Hamiltonian that transports $\su(3)$ elements locally
while conserving them globally. 
Such Hamiltonians may be engineered for superconducting qutrits,
as sketched in Sec.~\ref{sec_Outlook}.
% Furthermore, the Nambu--Jona-Lasinio (NJL) model,
% a precursor to QCD, has a similar form~\cite{Nambu_61_Dynamical,Vaks_60_On}.
However, this $\su(3)$ example only illustrates our more general prescription,
which works for all finite-dimensional semisimple complex Lie algebras 
on which the Killing form induces a metric.

Each basis for $\su(3)$ contains $\DimAlg = 8$ elements.
The most famous basis consists of the Gell-mann matrices,
$\lambda_{k = 1, 2, \ldots, 8}$~\cite{gell_10_symmetries}.
The $\lambda_k$ generalize the Pauli matrices in certain ways,
being traceless and Killing-orthogonal.
From the Gell-mann matrices is constructed 
the conventional Cartan-Weyl basis~\cite{Cahn_06_Semi}, 
reviewed in the Supplementary Note 5.
The $r = 2$ Hermitian elements are Gell-mann matrices:
\begin{equation}
   \label{eq_su3_firstcharges}
    Q_1 = \lambda_3,
     \quad \text{and} \quad 
     Q_2 = \lambda_8 .
\end{equation}
$Q_1$ and $Q_2$ belong in the preferred basis of charges for $\su(3)$.
For pedagogical clarity, we will identify all the charges
before addressing the ladder operators.

A general element $U \in$ SU(3) contains eight real parameters.
In the Euler parameterization~\citep{byrd1998differential},
\begin{align}
    \label{eq_euler_u3}
    U & = 
    e^{i\lambda_3 \phi_1/2} 
    e^{i\lambda_2 \phi_2/2}
    e^{i\lambda_3 \phi_3/2}
    e^{i\lambda_5 \phi_4/2}
    \nonumber \\ & \quad \times
    e^{i\lambda_3 \phi_5/2}
    e^{i\lambda_2 \phi_6/2}
    e^{i\lambda_3 \phi_7/2}
    e^{i\lambda_8 \phi_8/2} \, .
\end{align}
The parameters $\phi_1, \phi_3, \phi_5, \phi_7 \in [0,  2\pi)$; 
$\phi_2, \phi_4 ,\phi_6 \in [0,  \pi]$; and 
$\phi_8 \in [0,  2\sqrt{3}\pi)$.
We now constrain $U$, identifying the instances $U_{{\rm i}}$ that map 
the first charges to 
$Q_3 = U_{{\rm i}}^\dag Q_1 U_{{\rm i}}$ and $Q_4 = U_{{\rm ii}}^\dag Q_2 U_{\rm ii}$
that are Killing-orthogonal to each other and to the original charges. Supplementary Note 5 contains the details.
We label with a superscript $({\rm i})$ the parameters used to fix $U_{\rm i}$:
$\phi_1^{({\rm i})}$, $\phi_3^{({\rm i})}$, $\phi_7^{({\rm i})}$, $\phi_8^{({\rm i})}$, and $n^{({\rm i})}$.
For convenience, we package several parameters together:
$a^{({\rm i})} := \frac{1}{2} \left(\phi_3^{({\rm i})}
    - \phi_7^{({\rm i})}-\sqrt{3} \phi_8^{({\rm i})} + \pi n^{({\rm i})} + \tfrac{\pi}{2}\right)$,
and $b^{({\rm i})} := a^{({\rm i})} + \phi_7^{({\rm i})}$.
In terms of these parameters, the new charges have the forms (see Supplementary Note 5)
\begin{align}
   \label{eq_su3_gen_charge3}
    Q_3 & = \frac{1}{\sqrt{3}}
    \Big[ (-1)^{ n^{({\rm i})} + 1} \sin(a^{({\rm i})}  - b^{({\rm i})}) \lambda_1  
             \\ \nonumber & \quad
             -  (-1)^{n^{({\rm i})}}\cos(a^{({\rm i})} 
             - b^{({\rm i})}) \lambda_2
             -\sin(a^{({\rm i})}) \lambda_4  
             \\ \nonumber & \quad
             - \cos(a^{({\rm i})}) \lambda_5 + \sin(b^{({\rm i})})\lambda_6 
             + \cos(b^{({\rm i})})\lambda_7  \Big] 
    % % % 
    \; \; \text{and}
    % % %
    \\ \label{eq_su3_gen_charge4}
    Q_4 & = \frac{(-1)^{n^{({\rm i})}}}{\sqrt{3}} 
    \Big[  (-1)^{n^{({\rm i})} + 1}  \cos(a^{({\rm i})} - b^{({\rm i})}) \lambda_1  
             \\ \nonumber & \quad
            + (-1)^{n^{({\rm i})}}\sin(a^{({\rm i})} - b^{({\rm i})})\lambda_2 
            + \cos(a^{({\rm i})})\lambda_4 
            \\ \nonumber & \quad
            - \sin(a^{({\rm i})})\lambda_5 
            + \cos(b^{({\rm i})})\lambda_6 
            - \sin(b^{({\rm i})})\lambda_7\Big] .
\end{align}
$Q_3$ has the same form as $Q_5$ and $Q_7$,
which satisfy the same Killing-orthogonality conditions.
Similarly, $Q_4$ has the same form as $Q_6$ and $Q_8$.
The later charges' parameters $a^{(\ell)}$ and $b^{(\ell)}$ 
are more restricted, however (see Supplementary Note 5).
We have identified our preferred basis of charges.

Let us construct the ladder operators and Hamiltonian.
Each Cartan-Weyl basis contains $\DimAlg - r = 8 - 2 = 6$ ladder operators.
The conventional Cartan-Weyl basis contains ladder operators
formed from Gell-man matrices:
\begin{align}
   \label{eq_su3_Ladders_13}
   & L_{\pm1} := \tfrac{1}{2}(\lambda_1 \pm i \lambda_2),
   \quad
   L_{\pm 2} := \tfrac{1}{2}(\lambda_4 \pm i\lambda_5), 
   \nonumber \\ &  \; \text{and} \; \;
   L_{\pm 3}  := \tfrac{1}{2}(\lambda_6 \pm i\lambda_7).
\end{align}
Transforming these operators with unitaries $U_{\rm ii, iii, iv}$ yields 
$L_{\pm 4}$ through $L_{\pm 12}$,
whose forms appear in the Supplementary Note 5.
From each ladder operator, we form one term in 
the two-body Hamiltonian~\eqref{eq_H_2body_1}.

Finally, we determine the hopping frequencies $J_\alpha^{(j, j')}$,
demanding that $[H^{(j, j')}, Q_\alpha^\tot] = 0$ for all $\alpha$.
For all possible values of the $a^{(\ell)}$, $b^{(\ell)}$, and $n^{(\ell)}$,
if all the frequencies are nonzero,
then all the frequencies equal each other.
We set $J_\alpha^{(j, j')}  \equiv  \frac{4}{3} \, J^{(j, j')}$, such that
\begin{align}
   \label{eq_su3_H}
   H^{(j, j')}
   =  J^{(j, j')}  \sum_{\alpha = 1}^8  
   \lambda_\alpha^\JParen
   \lambda_\alpha^{(j')}
   % % %
   \propto   \sum_{\alpha = 1}^8
   Q_\alpha^\JParen
   Q_\alpha^{(j')} .
\end{align}
The Hamiltonian collapses to 
a simple form analogous to the $\su(2)$ example's Eq.~\eqref{eq_su2_H_Simple}
(see Supplementary Note 3).

\section{Outlook}
\label{sec_Outlook}

We have presented a prescription for constructing Hamiltonians that transport noncommuting charges locally while conserving the charges globally. The Hamiltonians can couple arbitrarily many subsystems together and can be integrable or nonintegrable. The prescription produces, as well as Hamiltonians, preferred bases of charges that are (i) overtly transported locally and conserved globally and (ii) Killing-form-orthogonal. This construction works whenever the charges form a finite-dimensional semisimple complex Lie algebra on which the Killing form induces a metric. Whether there exists any Hamiltonians that transport charges locally, while conserving the charges globally, outside of those found by our prescription, is an interesting open question for theoretical exploration.

This work provides a systematic means of bridging noncommuting thermodynamic charges from abstract quantum information theory to condensed matter, AMO physics, and high-energy and nuclear physics. The mathematical results that have accrued~\cite{Lostaglio_14_Masters, Guryanova_16_Thermodynamics, Lostaglio_17_Thermodynamic, NYH_18_Beyond, NYH_16_Microcanonical,Ito_18_Optimal, Bera_19_Thermo, Mur_Petit_18_Revealing, Gour_18_Quantum, Popescu_18_Quantum, Manzano_18_Squeezed,  NYH_20_Noncommuting,  Manzano_20_Non, Sparaciari_18_First, Khanian_20_From,  Khanian_20_Resource, Manzano_20_Hybrid, Fukai_20_Noncommutative, Mur-Petit_19_Fluctuations, Scandi_18_Thermodynamic, Popescu_19_Reference, Manzano_18_Squeezed, Sparaciari_18_First, Mur_Petit_18_Revealing, Popescu_18_Quantum, Boes_18_Statistical, Ito_18_Optimal, Gour_18_Quantum, Mitsuhashi_21_Characterizing} 
can now be tested experimentally, via our construction. 
This paper's introduction highlights example results that merit testing.
Such experiments' benefits include the simulation of quantum systems larger than what classical computers can simulate, the uncovering of behaviors not predicted by theory, and the grounding of abstract QIT thermodynamics in physical reality. 

In addition to harnessing controlled platforms to study
noncommuting charges' quantum thermodynamics,
one may leverage that quantum thermodynamics 
to illuminate high-energy and nuclear physics. Such physics includes non-Abelian gauge theories, such as quantum chromodynamics. How to define and measure such theories' thermalization
is unclear~\cite{Mueller_21_Thermalization}.
One might gain insights by using our dynamics as a bridge 
from quantum thermodynamics to non-Abelian field theories.

As mentioned above, the Heisenberg model~\eqref{eq_SU2_Ham_1}
can be implemented with ultracold atoms and trapped ions~\cite{Jane_03_Simulation, Barredo_16_Atom, de_19_observation, Zhang_17_Observation,Fukuhara_13_Microscopic, Viola_99_Universal}.
Reference~\cite{NYH_20_Noncommuting} details how to harness these setups 
to study noncommuting thermodynamic charges.
We introduce a more complex example here:
We illustrate, with superconducting qubits, how today's experimental platforms can implement the $\su(3)$ instance of our general prescription. 

Superconducting circuits can serve as qudits with Hilbert-space dimensionalities $\Dim \geq 2$~\cite{You_11_Atomic}. Qutrits have been realized with transmons, slightly anharmonic oscillators~\cite{koch_07_charge}. The lowest two energy levels often serve as a qubit, but the second energy gap nearly equals the first. Hence the third level can be addressed relatively easily~\cite{bianchetti_10_control}. Superconducting qutrits offer a tabletop platform for transporting and conserving $\su(3)$ charges as in Sec.~\ref{sec_su3}.

Experiments with $\leq 5$ qutrits have been run~\cite{Morvan_20_Qutrit, blok_20_quantum}, Furthermore, many of the tools used to control and measure superconducting qubits can be applied to qutrits~\cite{bianchetti_10_control, xu_16_coherent, kumar_16_stimulated,tan_18_topological,vepsalainen_19_superadiabatic,lu_17_nonleaky,vepsalainen_16_quantum,yang_12_generation,shlyakhov_18_quantum,danilin_18_experimental,shnyrkov_12_quantum}.
A noncommuting-charges-in-thermodynamics experiment may begin with preparing the qutrits in an approximate microcanonical subspace, a generalization of the microcanonical subspace that accommodates noncommuting charges~\cite{NYH_16_Microcanonical}. Such a state preparation may be achieved with weak measurements~\cite{NYH_20_Noncommuting}, which have been performed on superconducting qudits through cavity quantum electrodynamics~\cite{Naghiloo_19_Introduction}.
% Do not erase: One reason why we might expect the weak-measurement scheme to work: The light that pours out of the cavity carries information about all the qubits. We don't have single-qubit resolution, unless we work hard to gain it. So we should be able to measure the global charges effectively.

% Reference: email chain "SC qutrit decoherence times" -- messages sent on 3/12/21
$T_2^{*}$ relaxation times of $\sim 39 \; \mu$s, for the lowest energy gap, and $\sim 14 \; \mu$s, for the second-lowest gap, have been achieved~\cite{blok_20_quantum}. Meanwhile, two-qutrit gates can be realized in $\sim 10 - 10^2$ ns~\cite{blok_20_quantum,huang_20_superconducting, kjaergaard_20_superconducting}.
% References: (1) blok_20_quantum -- The 0 and 2 states of the one qutrit control the 1 state of the other qubit.
% (2) huang_20_superconducting, kjaergaard_20_superconducting -- Two levels of one qutrit coupled to two levels of another
% Operation time for 1-qutrit gate: 10-30 ns
Some constant number of such gates may implement one three-level gate that simulates a term in our Hamiltonian. If the number is order-10, information should be able to traverse an 8-qutrit system $\sim 10$ times before the qutrits decohere detrimentally. 
According to numerics in \citep{NYH_16_Microcanonical}, a small subsystem nears thermalization once information has had time to traverse the global system a number of times linear in $\Sites$. Therefore, realizations of our Hamiltonians are expected to thermalize the system internally. The states of small subsystems, such as qutrit pairs, can be read out via quantum state tomography~\cite{bianchetti_10_control, kumar_16_stimulated, xu_16_coherent, tan_18_topological, vepsalainen_19_superadiabatic}. Hence superconducting qutrits, and other platforms, can import noncommuting charges from quantum thermodynamics to many-body physics, by simulating the Hamiltonians constructed here.

\begin{acknowledgments}
NYH is grateful to Michael Beverland, Aram Harrow, Iman Marvian, Mark Mueller, and Martin Savage for thought-provoking conversations.
SSM would like to thank Jos\'e Polo G\'omez; Jimmy Shih-Chun Hung; Eduardo Mart\'in-Mart\'inez; Erickson Tjoa; and, in particular, Tibra Ali for fruitful discussions.
This work was supported by an NSF grant for the Institute for Theoretical Atomic, Molecular, and Optical Physics at Harvard University and the Smithsonian Astrophysical Observatory, as well as by administrative support from the MIT CTP. This work received support from the National Science Foundation (QLCI grant OMA-2120757).
\end{acknowledgments}

\section*{Author Contributions} 

NYH developed the prescription, managed the project, and led the paper writing. SM worked out the su(3) example, proofs, Supplementary Notes, and superconducting-qutrit details, in addition to leading the referee revisions.

\begin{appendices}

% To switch from two-column to one-column formatting 
\onecolumngrid

% Number subsections in the appendices as in the main text,
% except skip the capital Roman numerals.
\renewcommand{\thesection}{\Alph{section}}
\renewcommand{\thesubsection}{\Alph{section} \arabic{subsection}}
\renewcommand{\thesubsubsection}{\Alph{section} \arabic{subsection} \roman{subsubsection}}

% Label the equations in Appendix L as L1, L2, ...
\makeatletter\@addtoreset{equation}{section}
\def\theequation{\thesection\arabic{equation}}

\section{The Killing form induces a metric on every simple Lie algebra.}
\label{app_Killing_Metric}

Here, we prove a claim made in Sec.~\shayan{II.A of the main text}:
The Killing form induces a metric on every simple Lie algebra.
The proof relies on background material reviewed in Sec.~\shayan{II.B of the main text}.

Every inner product defines a metric.
Therefore, proving that the Killing form induces an inner product suffices.
On a simple Lie algebra, all symmetric bilinear forms equal each other
to within a multiplicative constant~\cite{humphreys_12_introduction}.
The Killing form is one symmetric bilinear form;
another is $\Tr ( Q_\alpha  Q_\beta )$.
Hence $(Q_\alpha,  Q_\beta)
\propto  \Tr ( Q_\alpha  Q_\beta ) 
=  \Tr ( Q_\alpha^\dag  Q_\beta )$.
The final equality follows from the charges' Hermiticity.
The final expression is the Hilbert-Schmidt inner product.
Hence the Killing form induces an inner product.

\section{General Hamiltonian that transports 
\texorpdfstring{$\mathfrak{su}(2)$}{su(2)} elements locally while conserving them globally}
\label{app_su2full}

Section~\shayan{II.B of the main text} illustrated how to construct Hamiltonians
that transport $\su(2)$ elements locally while conserving them globally.
The illustration was not maximally general;
we restricted a unitary $U$ more than required, for pedagogy.
We generalize the construction here.
For clarity of presentation, we derive the charges' forms first 
(\shayan{Supplementary Note} ~\ref{app_Basis_su2})
and the ladder operators' forms second (\shayan{Supplementary Note } ~\ref{app_Ladder_su2}).
We then construct the two-body Hamiltonian $H^{(j, j')}$
and a three-body Hamiltonian (\shayan{Supplementary Note} ~\ref{app_23Body_H_su2}).

\subsection{Preferred basis of charges for \texorpdfstring{$\su(2)$}{su(2)}
\label{app_Basis_su2}}

The conventional Cartan-Weyl basis contains the Hermitian operator
\begin{align}
   \label{eq_Q1_su2_app}
   Q_1 = \sigma_z .
\end{align}
To identify the next Cartan-Weyl basis, we invoke a general unitary
$U \in$ SU(2). In the Euler parameterization,
\begin{align}
   \label{eq:su2_gen_U}
    U &= e^{i\sg_z \phi_1/2}e^{i\sg_y \phi_2/2}e^{i\sg_z \phi_3/2} ,
\end{align}
wherein $\phi_1 \in [0,2\pi)$, $\phi_2 \in [0,\pi]$, and $\phi_3 \in [0,2\pi)$. 
We restrict this general unitary to a $U_{{\rm i}}$
that maps $Q_2$ to a Killing-orthogonal charge 
$Q_2 = U_{{\rm i}}^{\dagger}Q_1U_{{\rm i}}$.
For $X, Y \in \su(\SUD)$, the Killing form evaluates to
$(X, Y) = \Tr ( X Y )$~\cite{humphreys_12_introduction}.
Hence the Killing form between the charges is
\begin{align}
    0 = \left(  U_{{\rm i}}^{\dagger}Q_1U_{{\rm i}}, \, Q_1  \right) 
    = \Tr \left(  U_{{\rm i}}^{\dagger}Q_1U_{{\rm i}}Q_1  \right)
    = 2\cos \left( \phi_2^{({\rm i})}  \right) 
    \label{eq:su2KO1} .
\end{align}
The superscript $({\rm i})$, here and below, labels a parameter as belonging to $U_{{\rm i}}$.
The equation, with $\phi_2^{({\rm i})} \in [0,\pi]$, implies that $\phi_2^{({\rm i})} = \pi/2$. 
The unitary and charge assume the forms
\begin{align}
     U_{{\rm i}} = e^{i\sg_z \phi_1^{({\rm i})}/2}e^{i\sg_y \pi/4}e^{i\sg_z \phi_3^{({\rm i})}/2} 
%     \label{eq:su2_U2}
     \quad \text{and} \quad
    Q_2 = \cos(\phi_3^{({\rm i})}) \sg_x + \sin(\phi_3^{({\rm i})}) \sg_y .
    \label{eq:su2_Q2}
\end{align}

Having identified the second charge, we identify the final one.
We transform $Q_1$ with a unitary $U_{{\rm ii}} \in$ SU(2) such that 
$Q_3 = U_{{\rm ii}}^\dag  Q_1  U_{{\rm ii}}$ is Killing-orthogonal to the first two charges.
The first orthogonality constraint has the form of Eq.~\eqref{eq:su2KO1},
except that a $({\rm ii})$ replaces the superscript $({\rm i})$.
The second orthogonality constraint is
\begin{align}
    0 =  \Tr \left( U_{{\rm ii}}^\dag  Q_1  U_{{\rm ii}}  ,  \,  Q_2  \right)
    = \Tr  \left(  U_{{\rm ii}}^{\dagger}Q_1U_{{\rm ii}}  Q_2  \right) 
    = 2\cos \left( \phi_3^{({\rm i})} - \phi_3^{({\rm ii})} \right) .
\end{align}
Hence $\phi_3^{({\rm ii})} = \phi_3^{({\rm i})} + \pi \left( n^{({\rm ii})} - \frac{1}{2}  \right)$,
wherein $n^{({\rm ii})} \in \mathbb{Z}$. 
Hence $U_{{\rm ii}}$ and  $Q_3$ have the forms
\begin{align}
    U_{{\rm ii}} &=
    e^{i\sg_z \phi_1^{({\rm ii})}/2}e^{i\sg_y \pi/4}e^{i\sg_z [\phi_3^{({\rm i})}+\pi(n-\frac{1}{2})]/2} 
    \label{eq:su2_U3}
    \quad \text{and} \\
    Q_3 &= 
    (-1)^{n^{({\rm ii})}} \left[  \sin(\phi_3^{({\rm i})}) \sg_x -\cos(\phi_3^{({\rm i})}) \sg_y  \right] 
    \label{eq:su2_Q3} .
\end{align}
Equations~\eqref{eq:su2_Q3},~\eqref{eq:su2_Q2}, and~\eqref{eq_Q1_su2_app}
specify the preferred basis of charges for $\su(2)$.

\subsection{General ladder operators for \texorpdfstring{$\su(2)$}{su(2)}
\label{app_Ladder_su2}}

The conventional Cartan-Weyl basis contains operators that
raise and lower $\sigma_z$:
 \begin{align}
  \label{eq_su2_Ladder1_app}
  L_{\pm 1} 
  = \sg_{\pm z} 
  = \frac{1}{2} (\sg_x \pm i \sg_y).
\end{align}
Conjugation with $U_{{\rm i}}$ yields the ladder operators for $Q_2$,
and conjugation with $U_{{\rm ii}}$ yields the ladder operators for $Q_3$:
\begin{align}
   \label{eq_su2_Ladder23_app}
    L_{\pm 2} 
    &=  U_{{\rm i}}^\dag  L_{\pm 1}  U_{{\rm i}}
    = \frac{- e^{\mp i \phi_1^{({\rm i})}}}{2} 
    [\sg_z \pm i( \sin\{\phi_3^{({\rm i})}\}\sg_x - \cos\{\phi_3^{({\rm i})}\}\sg_y)] ,
    \quad \text{and} \\
    % % %
    L_{\pm 3}
    &=  U_{{\rm ii}}^\dag  L_{\pm 1}  U_{{\rm ii}}
    = \frac{- e^{\mp i \phi_1^{({\rm ii})}}}{2}   \,
    \left\{  \sg_z \mp i (-1)^{n^{({\rm ii})}}  
    \left[  \cos  \left(  \phi_3^{({\rm i})}  \right)  \sg_x  
            + \sin  \left(  \phi_3^{({\rm i})}  \right)  \sg_y  
    \right] \right\} .
\end{align}

\subsection{Two-body and three-body Hamiltonians for \texorpdfstring{$\su(2)$}{su(2)}
\label{app_23Body_H_su2}}

To form $H^{(j, j')}$, we substitute for the ladder operators 
from Eqs.~\eqref{eq_su2_Ladder1_app}
and~\eqref{eq_su2_Ladder23_app} into Eq.~\shayan{(12)}.
We require that $H^{(j, j')}$ conserve each global charge,
imposing Eq.~\shayan{(14)}.
This equation holds, algebra reveals, if and only if
the hopping frequencies $J_\alpha^{(j, j')}$ equal each other.
The Hamiltonian simplifies to Eq.~\shayan{(15)}.
The final expression does not depend on our choice of 
$\phi_k^{({\rm i})}$, $\phi_k^{({\rm ii})}$, or $n^{({\rm i})}$. 

Let us construct a Hamiltonian $H^{(j, j', j'')}$ that 
transfers $\su(2)$ charges between three sites---$j$, $j'$, and $j''$---while 
conserving the charges globally.
We multiply three two-body Hamiltonians together cyclically:
\begin{align}
   \label{eq_3body_H_app}
   H^{(j, j', j'')}
   % % %
   & \propto H^{(j, j')}  H^{(j', j'')}  H^{(j'', j)}
\end{align}
We substitute in from Eq.~\shayan{(13)},
the $H^{(j, j')}$ expression in which the hopping frequencies have not yet been restricted. 
The frequencies can assume different values,
when $[H^{(j, j', j'')},  \, Q_\alpha^\tot] = 0$, than when 
$[H^{(j, j')}, \, Q_\alpha^\tot] = 0$.
Imposing the first commutator equation yields four sets of 
solutions for the $J_\alpha$'s, when $J_\alpha \neq 0$ for all $\alpha$:
\begin{enumerate}

    \item \label{item_J_alpha_1}
    $J_1 = J_2 = J_3$, $J_4 = J_5 = J_6$, and $J_7 = J_8 = J_9$.
    
    \item
    $J_1 = J_2 = -J_3$, $J_4 = J_5 = -J_6$, and $J_7 = J_8 = -J_9$.

    \item 
    $J_1 = J_2 = \tfrac{-J_3}{2}$, 
    $J_4 = J_5 = \tfrac{-J_6}{2}$, and 
    $J_7 = J_8 =  \tfrac{-J_9}{2}$.
    
    \item 
    $\tfrac{J_2}{J_1} = \tfrac{J_5}{J_4} 
      = \tfrac{J_8}{J_7}$, $J_1 + J_2 
      = - J_3$, $J_4 + J_5 = - J_6$, and 
    $J_7 + J_8 = - J_9$.
\end{enumerate}
We have omitted superscripts for conciseness.
The four solutions lead to distinct Hamiltonians.\footnote{
% < f >
However, each solution contains a little redundancy:
Consider picking one of the four solutions,
then cycling the indices in $(1,2,3)$ identically to 
the indices in $(4,5,6)$ and to the indices in $(7,8,9)$.
The resulting $J_\alpha$'s specify a Hamiltonian
identical to the original.}
% < /f >

For concreteness, we detail the first set of solutions, item~\ref{item_J_alpha_1}.
We collect three of the frequencies to simplify notation:
$J^{j,j',j''} = J_1^{(j,j')}J_4^{(j',j'')}J_7^{(j',j'')}$.
Substituting the $J_\alpha$'s into the Hamiltonian~\eqref{eq_3body_H_app} yields
\begin{align}
    H^{(j, j', j'')}
    & \propto J^{(j,j',j'')}  \Big\{  3\id\id\id 
   - 2 \left(  H^{(j, j')}  -  H^{(j'', j)}  +  H^{(j', j'')}  \right) 
   \nonumber \\ & \qquad \qquad \quad \;
   + i [(\sg_x\sg_y\sg_z + \sg_y\sg_z\sg_x + \sg_z\sg_x\sg_y) 
    - ( \sg_z\sg_y\sg_x + \sg_x\sg_z\sg_y + \sg_y\sg_x\sg_z) \} .
\end{align}
We have omitted some superscripts to simplify notation.
The first term is trivial, terms 2-4 are two-body,
and each of terms 1-4 conserves each $Q_\alpha^\tot$.
Subtracting these terms off yields the solely three-body Hamiltonian, \shayan{Eq.~(18)}.
We have absorbed the $i$ into the coefficient such that 
$J^{j,j',j''} \in \mathbb{R}$.

\section{Simple form to which a two-body Hamiltonian may collapse}
\label{app_proof_thm_Struc_Consts}

In the $\su(2)$ example, $H^{(j, j')}$ collapsed to 
the simple form~\shayan{(16)}.
The $\su(3)$ $H^{(j, j')}$ collapses to an analogous form,
we shown in Sec.~\shayan{II.D}.
This form generalizes to
\begin{align}
   \label{eq_H_Simple_Form}
   \sum_{\alpha = 1}^c
   Q_\alpha^\JParen  Q_\alpha^{(j')} .
\end{align}
This expression generally conserves noncommuting charges globally,
and transport the charges locally, as proved below.
However, the expression's equality with a two-body Hamiltonian that
clearly, overtly transports local charges from site to site
is proved only in the $\su(2)$ and $\su(3)$ examples.

% Proposition
\begin{proposition}
\label{prop_H_Simple_Form}
Consider any Lie algebra whose structure constants have 
the antisymmetry property
\begin{align}
   \label{eq_Antisym_Struc_Consts}
   f_{\alpha \beta}^\gamma 
   = - f_{\gamma \beta}^\alpha .
\end{align}
A two-body Hamiltonian of the form~\eqref{eq_H_Simple_Form}
conserves the algebra's elements globally.
\end{proposition}
\noindent 
Every compact semisimple Lie algebra has such structure constants~\cite{metha_83_property}.

% Proof
\begin{proof}
First, we substitute from Eq.~\eqref{eq_H_Simple_Form}
into the conservation law.
Then, we invoke the commutator's linearity and 
the arguments' tensor-product forms:
\begin{align}
   0 & = \left[  H^{(j, j')},  \,  Q_\alpha^\tot  \right] 
   % % %
   = \left[  \sum_{\beta = 1}^c  Q_\beta^\JParen  Q_\beta^{(j')},  \,
   Q_\alpha^\JParen  \otimes \id^{(j')}
   +  \id^{(j)}  \otimes  Q_\alpha^{(j')}  \right] \\
   % % %
   & = \sum_{\beta = 1}^c  \left(
   \left[  Q_\beta^\JParen  Q_\beta^{(j')} ,  \,
            Q_\alpha^{(j)}  \otimes  \id^{(j')}  \right]
   +  \left[  Q_\beta^\JParen  Q_\beta^{(j')},  \, 
                \id^\JParen  \otimes
                Q_\alpha^{(j')}  \right]
   \right) \\
   % % %
   \label{eq_Struc_Const_Help1}
   & = \sum_{\beta = 1}^c  \left(
   \left[  Q_\beta^\JParen ,  \,  Q_\alpha^\JParen  \right]
   Q_\beta^{(j')}
   +  Q_\beta^\JParen
       \left[  Q_\beta^{(j')} ,  \,  Q_\alpha^{(j')}  \right]
   \right) .
\end{align}
Let $f_{\alpha \beta}^\gamma$ denote the Lie algebra's structure constants.
The $f$'s dictate how a Lie bracket decomposes
as a linear combination of the algebra's elements:
\begin{align}
   [ Q_\alpha,  \,  Q_\beta ]
   % % %
   & = \sum_{\gamma = 1}^c  
   f_{\alpha \beta}^\gamma
   Q_\gamma .
\end{align}
We substitute into Eq.~\eqref{eq_Struc_Const_Help1}, 
then pull the sums and constants out front:
\begin{align}
   0
   % % %
   & =  \sum_{\beta = 1}^c  \left[
   \left(  \sum_{\gamma = 1}^c  f_{\beta \alpha}^\gamma  
            Q_\gamma^\JParen  \right)
   Q_\beta^{(j')}
   +  Q_\beta^\JParen
   \left(  \sum_{\gamma = 1}^c
            f_{\beta \alpha}^\gamma
            Q_\gamma^{(j')}  \right)
   \right] 
   % % %
   =  \sum_{\beta, \gamma = 1}^c
   f_{\beta \alpha}^\gamma
   \left(  Q_\gamma^\JParen  Q_\beta^{(j')}
            +  Q_\beta^\JParen  Q_\gamma^{(j')}  \right) .
\end{align}
The final equation holds if 
$f_{\beta \alpha}^\gamma =  - f_{\gamma \alpha}^\beta$.
Consider relabeling the index $\alpha$ as $\beta$ and vice versa.
Equation~\eqref{eq_Antisym_Struc_Consts} results.
\end{proof}

Having proved that the simple operator~\eqref{eq_H_Simple_Form} conserves noncommuting charges globally, we prove that it transports charges locally.

% Proposition
\begin{proposition}
The simple two-body Hamiltonian~\eqref{eq_H_Simple_Form} 
transports the charges $Q_\alpha$ locally.
\end{proposition}

% Proof
\begin{proof}
Charge $Q_\alpha$ is transported locally if it satisfies Eq.~\shayan{(3)},
having a nonzero commutator
\begin{align}
  \comm{H^{(j,j')}}{Q_{\alpha}^{(j)}}  
  % % %
  =  \comm{   
  \sum_{\beta = 1}^c   Q_\beta^\JParen  Q_\beta^{(j')}  }{  Q_{\alpha}^{(j)}  } 
  % % %
  = \sum_{\beta = 1}^c  \left[  
  Q_\beta^\JParen,   Q_\alpha^\JParen  \right]
  Q_\beta^{(j')}  
  % % %
  \label{eq_Simple_Ham_Help1}
  =  \sum_{\beta, \gamma = 1}^c
  f_{\beta \alpha}^\gamma  \,
  Q_\gamma^\JParen
  Q_\beta^{(j')} .
\end{align}
The final expression vanishes if $Q_\alpha$ commutes with
all the other charges $Q_\gamma$ in the preferred basis.
If a Lie algebra has a basis of which one element commutes with the others,
the algebra is Abelian, by definition~\cite{humphreys_12_introduction}.
We assume that the algebra $\Algebra$ is non-Abelian (Sec.~\shayan{II.A of the main text}).
Therefore, the right-hand side of~\eqref{eq_Simple_Ham_Help1} is nonzero,
and the Hamiltonian transports the charges locally.
\end{proof}

\section{Proof of Proposition 1}
\label{sec_Alg_Props}

Proposition \shayan{1} states that 
the algebra $\Algebra$ has an integer ratio $c/r$,
wherein $c$ denotes the algebra's dimension and $r$ denotes the rank.

% Proof
\begin{proof}
For every finite-dimensional complex Lie algebra, 
there exists a corresponding connected Lie group that is 
unique to within finite coverings. 
The Lie algebra has the same dimension and rank as
each of the corresponding Lie groups. 
Thus, if Proposition \shayan{1} holds 
for all semisimple Lie groups, 
it holds for all semisimple Lie algebras.
We prove the group claim.

Every Lie group has a maximal torus $\mathbb{T}^{r}$,
which is the group generated by a Cartan subalgebra of the Lie algebra.
The torus' dimensionality equals the group's rank, $r$.
A torus is an $r$-fold Cartesian product of $\mathbb{S}^{1}$ manifolds 
[equivalently, of the group U(1)]. 
Quotienting out the torus' action from the Lie group yields
a finite-dimensional coset space.
Every finite-dimensional coset space's dimensionality is a positive integer
$n \in \mathbb{Z}_{>0}$.
Thus, the semisimple Lie group's dimension is $c = rn$.
\end{proof}

\section{Mathematical details: Construction of a two-body Hamiltonian that 
transports \texorpdfstring{$\mathfrak{su}(3)$}{su(3)} elements locally  
while conserving them globally} 
\label{app_math_details_su3}

Section~\shayan{II.D} illustrated the Hamiltonian-construction prescription with $\su(3)$.
We flesh out the explanation here. 
Appendix~\ref{app_su3_Cartan_Weyl} reviews 
the conventional Cartan-Weyl basis for $\su(3)$.
Appendix~\ref{app_sub_findq3q4} identifies 
the preferred basis of charges for $\su(3)$.
Appendix~\ref{app_su3_ladder_ops} presents the ladder operators
from which we construct a Hamiltonian.

\subsection{Conventional Cartan-Weyl basis for \texorpdfstring{$\su(3)$}{su(3)} 
\label{app_su3_Cartan_Weyl}}

$\mathfrak{su}$(3) has dimension $\DimAlg = 8$ and rank $r = 2$.
The conventional Cartan-subalgebra generators
are denoted by $t_z = \lambda_3 / 2$ and $y = \lambda_8 / \sqrt{3}$,
wherein $\lambda_3$ and $\lambda_8$ denote Gell-mann matrices~\cite{Cahn_06_Semi}.
% Reference: Cahn_84_Semi -- p. 16
These generators, in the three-dimensional representation 
of $\mathfrak{su}$(3), manifest as
\begin{align}
   T_z  =  \frac{1}{2}  \begin{bmatrix}
   1 & 0 & 0 \\
   0 & -1 & 0 \\
   0 & 0 & 0
   \end{bmatrix}
   % % %
   \quad \text{and} \quad
   Y  =  \frac{1}{3}  \begin{bmatrix}
   1 & 0 & 0 \\
   0 & 1 & 0 \\
   0 & 0 & -2
   \end{bmatrix}  .
\end{align}
$t_z$ and $y$ are orthogonal relative to the Killing form.
They (more precisely, rescaled versions of them) belong in our preferred basis of charges:
$Q_1 \propto t_z$, and $Q_2 \propto y$.

These charges are raised and lowered by
$\DimAlg - r = 8 - 2 = 6$ ladder operators,
% Reference: Cahn_84_Semi -- p. 16
$t_\pm = (\lambda_1 \pm i \lambda_2) / 2$,
$v_\pm = (\lambda_4 \pm i \lambda_5 ) / 2$, 
and $u_\pm = (\lambda_6 \pm i \lambda_7 ) / 2$.
In the three-dimensional representation of $\mathfrak{su}$(3),
the ladder operators manifest as
\begin{align}
   & T_+  =  \frac{1}{2}
   \begin{bmatrix}
      0 & 1 & 0 \\
      0 & 0 & 0 \\
      0 & 0 & 0
   \end{bmatrix} ,  \quad
   % % % 
   T_-  =  \frac{1}{2}
   \begin{bmatrix}
      0 & 0 & 0 \\
      1 & 0 & 0 \\
      0 & 0 & 0
   \end{bmatrix} ,  \quad
   % % %  
   V_+  =  \frac{1}{2}
   \begin{bmatrix}
      0 & 0 & 1 \\
      0 & 0 & 0 \\
      0 & 0 & 0
   \end{bmatrix} ,  \quad
   % % %
   V_-  =  \frac{1}{2}
   \begin{bmatrix}
      0 & 0 & 0 \\
      0 & 0 & 0 \\
      1 & 0 & 0
   \end{bmatrix} ,     \\ &
   % % %
   U_+  =  \frac{1}{2}
   \begin{bmatrix}
      0 & 0 & 0 \\
      0 & 0 & 1 \\
      0 & 0 & 0
   \end{bmatrix} , 
   % % %
   \quad \text{and} \quad
   U_-  =  \frac{1}{2}
   \begin{bmatrix}
      0 & 0 & 0 \\
      0 & 0 & 0 \\
      0 & 1 & 0
   \end{bmatrix} .
\end{align}

The ladder operators participate in the following commutation relations with the charges:
% Reference: Foster_16_First -- p. 3
\begin{align}
   \label{eq_Comm_Rels1}
   & [ t_z,  t_\pm ]  =  \pm t_\pm ,  \quad
   % % %
   [y, t_\pm ]  =  0 ,   \\ 
   % % %
   & [t_z,  v_\pm]  =  \pm \frac{1}{2} v_\pm ,  \quad
   % % %
   [y, v_\pm]  =  \pm v_\pm \\
   % % %
   \label{eq_Comm_Rels3}
   & [t_z, u_\pm]  =  \mp  \frac{1}{2} u_\pm ,  
   \quad \text{and} \quad
   % % %
   [y, u_\pm]  =  \pm u_\pm .
\end{align}
These relations imply that 
(i) $t_\pm$ raises and lowers $t_z$, whereas
(ii) $v_\pm$ raises or lowers both $t_z$ and $y$,
as does $u_\pm$.
We can prove this physical significance easily:
Let $L_\pm$ denote a ladder operator (a $t_\pm$, a $v_\pm$, or a $u_\pm$)
that raises/lowers a charge $Q$.
Let $\ket{ \psi }$ denote a $Q$ eigenstate
associated with the eigenvalue $q$:
$Q \ket{\psi}  =  q \ket{\psi}$.
Consider operating on the state with the ladder operator:
$L_\pm  \ket{\psi}$.
Suppose, for notational convenience, that, (i) if $L_+$ operates,
$q$ is not the greatest $Q$ eigenvalue and
(ii) if $L_-$ operates,
$q$ is not the least $Q$ eigenvalue.
The resulting state is a $Q$ eigenstate
associated with the eigenvalue $q \pm a$,
wherein $a = 1$ or $1/2$.
To prove this claim, we operate on the new state with the charge:
$Q ( L_\pm \ket{\psi} )$.
Invoking the appropriate commutation relation 
[Eqs.~\eqref{eq_Comm_Rels1}-\eqref{eq_Comm_Rels3}] yields
\begin{align}
   \label{eq_Comm_Rels4}
   Q L_\pm \ket{\psi}
   &=  (L_\pm  Q  \pm  L_\pm )  \ket{\psi}
   =  L_\pm  (Q  \pm  a \id)  \ket{\psi}
   =  L_\pm  (q  \pm  a)  \ket{\psi}
   =  (q  \pm  a)  L_\pm  \ket{\psi}.
\end{align}
By Eqs.~\eqref{eq_Comm_Rels1}-\eqref{eq_Comm_Rels4},
$t_\pm$ raises/lowers the $t_z$ charge by one quantum
and preserves $y$.
$u_\pm$ lowers/raises $t_z$ by half a quantum
and raises/lowers $y$ by one quantum.
$v_\pm$ raises/lowers each of $t_z$ and $y$
by one quantum.

Having reviewed the conventional Cartan-Weyl basis for $\su(3)$,
we dispense with the conventional notation ($t_z$, $t_\pm$, etc.).
We revert to the notation introduced in the main text
($Q_\alpha$ and $L_{\pm \alpha}$).

\subsection{Preferred basis of charges for \texorpdfstring{$\su(3)$}{su(3)}
\label{app_sub_findq3q4}}

The first two charges appear in Eqs.~\shayan{(20)}.
We construct two new charges from $Q_1$, $Q_2$, 
and a unitary $U \in$ SU(2).
The general form of such a $U$, appears, 
in the Euler parameterization, in Eq.~\shayan{(21)}.
We constrain $U$ with the Killing-orthogonality conditions~\shayan{(9)},
obtaining a unitary $U_{{\rm i}}$.
The transformed charges have the forms
$Q_3 = U_{{\rm i}}^\dag Q_1 U_{{\rm i}}$ and $Q_4 = U_{{\rm i}}^\dag Q_2 U_{{\rm i}}$.
The new charges are Killing-orthogonal to each other by unitarity:
$0 = \Tr  \left(  \left[  U_{{\rm i}}^{\dagger}Q_1U_{{\rm i}}  \right]  
                         \left[  U_{{\rm i}}^{\dagger}Q_2U_{{\rm i}}  \right] \right)
    = \Tr  \left(  Q_1Q_2  \right) = 0$.
Killing-orthogonality to the old charges,
Eq.~\shayan{(20)},
with the form of the $\su(\SUD)$ Killing form~\cite{humphreys_12_introduction}, implies
\begin{align}
     0 & = \Tr  \left(  \left[  U_{{\rm i}}^{\dagger}Q_1U_{{\rm i}}  \right]  Q_2  \right)
     = - \cos(\phi_2) / 3 ,    %&\rightarrow&  \cos(\phi_2) = 0,
     \qquad \qquad \qquad \quad \; \; \:
     % % %
     0 = \Tr  \left(  \left[  U_{{\rm i}}^{\dagger}Q_1U_{{\rm i}}  \right]  Q_1  \right) 
    = - \frac{1}{2\sqrt{3}}  \cos(\phi_3 + \phi_5) , \\
    %  &\rightarrow&   \cos(\phi_3 + \phi_5) = 0
    % % %
     0 & = \Tr  \left(  \left[  U_{{\rm i}}^{\dagger}Q_2 U_{{\rm i}}  \right]  Q_2  \right) 
     = \frac{1}{2} \left[  \cos(\phi_4) + \frac{1}{3}  \right]  ,
     %&\rightarrow&  \cos(\phi_4) = -\frac{1}{3},
     \qquad \text{and} \qquad
     % % %
    0 = \Tr  \left(  \left[  U_{{\rm i}}^{\dagger}Q_2U_{{\rm i}}  \right]  Q_1  \right) 
    = - \cos(\phi_6) / 3.  %&\rightarrow&   \cos(\phi_6) = 0,
\end{align}
Since $\phi_2 , \phi_4, \phi_6 \in [0, \pi]$ and 
$\phi_3, \phi_5 \in [0, 2\pi)$,
$\phi_2 = \frac{\pi}{2}$, 
$\phi_4 = \acos( - 1 / 3)$, 
$\phi_6 = \frac{\pi}{2}$ and 
$\phi_5 = \pi( n - 1/2 ) - \phi_3$, for 
$n \in \{1,2,3,4\}$.

Transforming $Q_1$ and $Q_2$ with a $U_{{\rm ii}} \in$ SU(3)
yields the charges $Q_5$ and $Q_6$,
and transforming $Q_1$ and $Q_2$ with a $U_{{\rm iii}} \in$ SU(3)
yields $Q_7$ and $Q_8$.
These last four charges are Killing-orthogonal to $Q_1$ and $Q_2$, 
like $Q_3$ and $Q_4$.
So $U_{{\rm ii}}$ and $U_{{\rm iii}}$ share the form of $U_{{\rm i}}$.
However, parameters $a^{({\rm ii})}$ and $b^{({\rm ii})}$, or $a^{({\rm iii})}$ and $b^{({\rm iii})}$,
replace the $a^{({\rm i})}$ and $b^{({\rm i})}$.
The later unitaries' parameters are more constrained than 
the $U_{{\rm i}}$ parameters.
Similarly, $Q_5$ through $Q_8$ share the forms of $Q_3$ and $Q_1$,
apart from their more-constrained parameters.

Evaluating the restrictions on all the charges simultaneously will prove useful. 
First, the conditions for $Q_5$ to be orthogonal to $Q_3$ and $Q_4$ are
\begin{align}
    \label{eq_Q5_Constraint_Q3}
    0  & =  \Tr (Q_5Q_3)
    \propto (-1)^{n^{({\rm i})} + n^{({\rm ii})}}  
    \cos(a^{({\rm i})} - a^{({\rm ii})} - b^{({\rm i})} + b^{({\rm ii})}) 
    + \cos(a^{({\rm i})} - a^{({\rm ii})})  + \cos(b^{({\rm i})} - b^{({\rm ii})})
    % % %
    \quad \text{and} \\
    % % %
    0  & =  \Tr (Q_5 Q_4) 
    \propto (-1)^{n^{({\rm i})}
    + n^{({\rm ii})}}\sin(a^{({\rm i})} - a^{({\rm ii})} - b^{({\rm i})} + b^{({\rm ii})}) 
    - \sin(a^{({\rm i})} - a^{({\rm ii})}) + \sin(b^{({\rm i})} - b^{({\rm ii})}).
\end{align}
The orthogonality conditions for $Q_6$ impose the same constraints, since 
$\Tr(Q_6Q_3)  \propto \Tr(Q_5Q_4)$ and $  \Tr (Q_6Q_4)  \propto \Tr(Q_5Q_3)$
(as can be checked explicitly).
Similarly, the orthogonality conditions on $Q_7$ evaluate to
\begin{align}
    \label{eq_Q7_Constraint_Q3}
    0  & =  \Tr ( Q_7 Q_3 )
    \propto (-1)^{n^{({\rm i})}+n^{({\rm iii})}}\cos(a^{({\rm i})} - a^{({\rm iii})} - b^{({\rm i})} + b^{({\rm iii})}) 
    + \cos(a^{({\rm i})} - a^{({\rm iii})}) + \cos(b^{({\rm i})} - b^{({\rm iii})}) ,  \\
    % % %
    0  & =  \Tr (Q_7 Q_4)
    \propto (-1)^{n^{({\rm i})}+n^{({\rm iii})}}\sin(a^{({\rm i})} - a^{({\rm iii})} - b^{({\rm i})} + b^{({\rm iii})})
    - \sin(a^{({\rm i})} - a^{({\rm iii})}) + \sin(b^{({\rm i})} - b^{({\rm iii})}) ,  \\
    % % %
    0  & =  \Tr (Q_7  Q_5)
    \propto  (-1)^{n^{({\rm ii})}+n^{({\rm iii})}}\cos(a^{({\rm ii})} - a^{({\rm iii})} - b^{({\rm ii})} + b^{({\rm iii})}) 
    + \cos(a^{({\rm ii})} - a^{({\rm iii})}) + \cos(b^{({\rm ii})} - b^{({\rm iii})}) ,
    \; \text{and} \\
    % % %
    \label{eq_Q7_Constraint_Q6}
    0  & =  \Tr (Q_7  Q_6)
    \propto (-1)^{n^{({\rm ii})}+n^{({\rm iii})}}\sin(a^{({\rm ii})} - a^{({\rm iii})} - b^{({\rm ii})} + b^{({\rm iii})})
    - \sin(a^{({\rm ii})} - a^{({\rm iii})}) + \sin(b^{({\rm ii})} - b^{({\rm iii})}).
\end{align}
The orthogonality conditions for $Q_8$ impose the same constraints 
[Eqs.~\eqref{eq_Q7_Constraint_Q3}-\eqref{eq_Q7_Constraint_Q6}].

We now identify sets of $a^{(\ell)}, b^{(\ell)},$ and $n^{(\ell)}$ 
that are solutions for all six constraints, 
Eqs.~\eqref{eq_Q5_Constraint_Q3}-\eqref{eq_Q7_Constraint_Q6}.
First, we define
$x_{\ell m} := a^{(\ell)} - a^{(m)}$ and 
$ y_{\ell m} := b^{(\ell)} - b^{(m)}$, for 
$(\ell, m) = (2,3), (2,4), (3,4)$. 
By these definitions, $x_{24} = x_{23} + x_{34}$, 
and $y_{24} = y_{23} + y_{34}$.
Second, the values of the $n^{(\ell)}$ themselves are irrelevant.
Only whether $n^{(\ell)} + n^{(m)}$ is even or odd matters.
Only four unique possibilities for the $n^{(\ell)}$ exist: 
All the $n^{(\ell)} + n^{(m)}$ are even; 
or one $n^{(\ell)} + n^{(m)}$ is even, 
while the other two sums are odd. 
A solution can therefore be expressed in terms of just four quantities:
$x_{23}, x_{34}, y_{23},$ and $y_{34}$. 
Each solution is periodic:
\begin{align}
   \label{eq_Period}
   ( x_{23}, x_{34}, y_{23}, y_{34} )
   \equiv ( x_{23}, x_{34}, y_{23}, y_{34} )
   + ( 2 \pi n, 2 \pi n, 2 \pi n, 2 \pi n ) ,
\end{align}
% The solutions are invariant up to an additive constant of $2\pi n$, 
wherein $n \in \mathbbm{Z}$. Therefore, we omit the $2\pi n$ when listing the solutions below.

First, suppose that all the $n^{(\ell)} + n^{(m)}$ are even. 
The constraints~\eqref{eq_Q5_Constraint_Q3}-\eqref{eq_Q7_Constraint_Q6} 
admit of 18 solutions. The first ten are
\begin{align}
    \left( x_{23},x_{34},y_{23},y_{34} \right) & =
    \left( 0, \pm \frac{2\pi}{3},  \mp \frac{2\pi}{3}, \pm \frac{2\pi}{3}  \right), \,
    \left( 0, 0,  \pm \frac{2\pi}{3}, \pm \frac{2\pi}{3} \right), \,
    \left( 0, \pm \frac{2\pi}{3}, \pm \frac{2\pi}{3}, 0 \right),  \,
    \left( \pm \frac{2\pi}{3}, 0, \pm \frac{2\pi}{3}, \mp \frac{2\pi}{3} \right), 
    \nonumber \\ & \quad \; 
    \left( \pm \frac{2\pi}{3}, \pm \frac{2\pi}{3}, \pm \frac{2\pi}{3}, \pm \frac{2\pi}{3} \right).
\end{align}
The next eight solutions are identical to the first eight, except that each $x_{\ell m}$ is swapped with the corresponding $y_{\ell m}$.

Second, $n^{({\rm i})} + n^{({\rm iii})}$ can be even while  $
n^{({\rm i})} + n^{({\rm ii})}$ and  $n^{({\rm ii})} + n^{({\rm iii})}$ are odd. 
The constraints~\eqref{eq_Q5_Constraint_Q3}-\eqref{eq_Q7_Constraint_Q6} admit of another 18 solutions. The first ten are 
\begin{align}
    \left( x_{23},x_{34},y_{23},y_{34} \right) &=
    \left( \pi, \pm \frac{\pi}{3},  \mp \frac{\pi}{3}, \pm \frac{\pi}{3} \right), \,
    \left( \pi, \pi,  \pm \frac{\pi}{3}, \pm \frac{\pi}{3} \right), \,
    \left( \pi, \pm \frac{\pi}{3},  \pm \frac{\pi}{3}, \pi \right),  \,
    \left( \pm \frac{\pi}{3}, \pi,  \pm \frac{\pi}{3}, \mp \frac{\pi}{3} \right), 
    \nonumber \\ & \quad \; 
    \left( \pm \frac{\pi}{3},  \pm \frac{\pi}{3}, \pm \frac{\pi}{3},\pm \frac{\pi}{3} \right).
\end{align}
The next eight solutions are identical to the first eight,
except that each $x_{\ell m}$ is swapped with the corresponding $y_{\ell m}$.

Third, $n^{({\rm i})} + n^{({\rm ii})}$ can be even while  
$n^{({\rm i})} + n^{({\rm iii})}$ and  $n^{({\rm ii})} + n^{({\rm iii})}$ are odd. 
The constraints~\eqref{eq_Q5_Constraint_Q3}-\eqref{eq_Q7_Constraint_Q6} admit of another 18 solutions. The first ten are
\begin{align}
    \left( x_{23},x_{34},y_{23},y_{34} \right) &=
    \left( 0, \pm \frac{\pi}{3},  \pm \frac{2\pi}{3}, \pm \frac{\pi}{3} \right), \,
    \left( 0, \pi,  \pm \frac{2\pi}{3}, \mp \frac{\pi}{3} \right), \,
    \left( 0, \mp \frac{\pi}{3}, \pm \frac{2\pi}{3}, \pi \right),  \,
    \left( \pm \frac{2\pi}{3}, \pi, \pm \frac{2\pi}{3}, \pm \frac{\pi}{3} \right), 
    \nonumber \\ & \quad \; 
    \left( \pm \frac{2\pi}{3}, \mp \frac{\pi}{3}, \pm \frac{2\pi}{3}, \mp \frac{\pi}{3} \right).
\end{align}
The next eight solutions are identical to the first eight, 
except that each $x_{\ell m}$ is swapped with the corresponding $y_{\ell m}$.

Fourth, suppose that $n^{({\rm ii})} + n^{({\rm iii})}$ is even while  
$n^{({\rm i})} + n^{({\rm ii})}$ and  $n^{({\rm i})} + n^{({\rm iii})}$ are odd.
The constraints~\eqref{eq_Q5_Constraint_Q3}-\eqref{eq_Q7_Constraint_Q6} admit of another 18 solutions. The first ten are
\begin{align}
    \left( x_{23},x_{34},y_{23},y_{34} \right) &=
    \left( \pi, \pm \frac{2\pi}{3},  \pm \frac{\pi}{3}, \pm \frac{2\pi}{3} \right), \,
    \left( \pi, 0,  \pm \frac{\pi}{3}, \mp \frac{2\pi}{3} \right), \,
    \left( \pi, \mp \frac{2\pi}{3}, \pm \frac{\pi}{3}, 0 \right),  \,
    \left( \pm \frac{\pi}{3}, 0, \pm \frac{\pi}{3}, \pm \frac{2\pi}{3} \right), 
    \nonumber \\ & \quad \; 
    \left( \pm \frac{\pi}{3}, \mp \frac{2\pi}{3}, \pm \frac{\pi}{3}, \mp \frac{2\pi}{3} \right).
\end{align}
The next eight solutions are identical to the first eight, except that each $x_{\ell m}$ is swapped with the corresponding $y_{\ell m}$.

One can check explicitly that the tuple 
$(x_{23} + y_{23}, \, x_{34} + y_{34})$ has three possible values:
$(x_{23} + y_{23}, \, x_{34} + y_{34}) = (\pm 2\pi/3, \, \pm 2\pi/3),
(\pm 4\pi/3, \, \pm 4\pi/3), (\pm 2\pi/3, \, \mp 4\pi/3)$. 
Three sets of solutions follow. For example, the first set of solutions is
$(x_{23} + y_{23}, \, x_{34} + y_{34}) = (\pm 2\pi/3, \, \pm 2\pi/3).$
Hence
\begin{align}
    & a^{({\rm i})} - a^{({\rm ii})} + b^{({\rm i})} - b^{({\rm ii})} 
    = \pm \frac{2\pi}{3} \, ,  \qquad \quad \; \;
    % % %
    a^{({\rm ii})} - a^{({\rm iii})} + b^{({\rm ii})} - b^{({\rm iii})} 
    = \pm \frac{2\pi}{3} \, , \\
    % % %
    & a^{(\ell)} - a^{(m)} \in \left\{  0, \frac{\pm \pi}{3},\frac{\pm 2\pi}{3}, \pi  \right\} \, ,
    \quad \text{and} \quad
    % % %
    b^{(\ell)} - b^{(m)} \in \left\{ 0, \frac{\pm \pi}{3}, \frac{\pm 2\pi}{3}, \pi  \right\} ,
\end{align}
for $(\ell,m) = (2,3)$ and $(3,4)$.
All the solutions lead to the same Hamiltonian, Eq.~\shayan{(25)}.
% Confirmed in email chain "A few minor comments on the draft" -- message sent by Shayan on 3/16/21

%
%
%
\subsection{Ladder operators for \texorpdfstring{$\mathfrak{su}(3)$}{su(3)} \label{app_su3_ladder_ops}}

The conventional Cartan-Weyl basis contains six ladder operators
[Eqs.~\shayan{(24)}].
We transform $L_{\pm 1,2,3}$ with the unitaries $U_{\rm i}$, $U_{\rm ii}$, and $U_{\rm iii}$ of \shayan{Sec.~\ref{app_sub_findq3q4}},
to construct the rest of the ladder operators:
$L_{\pm 4} = U_{{\rm i}}^{\dagger} L_{\pm1} U_{{\rm i}}$, 
$ L_{\pm 5} = U_{{\rm i}}^{\dagger} L_{\pm2} U_{{\rm i}}$, and 
$ L_{\pm 6} = U_{{\rm i}}^{\dagger} L_{\pm3} U_{{\rm i}}$.
Substituting in for $L_{\pm 1, 2, 3}$ from Eq.~\shayan{(24)} yields
\begin{align}
   L_{\pm 4} 
   & =  \frac{ie^{\mp i\phi_1^{({\rm i})}}}{6}
   \Big\{ 2i\cos( a^{({\rm i})} - b^{({\rm i})} )  \lambda_1 
            - 2i\sin(a^{({\rm i})} - b^{({\rm i})})\lambda_2  
            \mp \left[  \sqrt{3} \mp i(-1)^{n^{({\rm i})}}  \right]
                   \left[\cos( a^{({\rm i})}) \lambda _4 - \sin(a^{({\rm i})}) \lambda _5 \right]  
   \nonumber \\ & \qquad \qquad \quad \;
   \pm \left[  \sqrt{3} \pm i(-1)^{n^{({\rm i})}}  \right]  
          \left[  \cos(b^{({\rm i})}) \lambda _6 -  \sin(b^{({\rm i})}) \lambda _7  \right] 
   \mp \sqrt{3} (-1)^{n^{({\rm i})}} \lambda_3 
   - \sqrt{3}i\lambda_8  \Big\} ,  \\
   % % %
   L_{\pm 5} & = \frac{  ie^{\mp\frac{i}{2}  (  \phi_3^{({\rm i})} + \phi_1^{({\rm i})}  )  }}{6}
   \Big( i  \left[  \cos(a^{({\rm i})}- b^{({\rm i})}) 
            - \sqrt{3}(-1)^{n^{({\rm i})}}  \sin(a^{({\rm i})} - b^{({\rm i})})  \right]  \lambda_1  
   \nonumber \\ & \qquad
    -i  \left[  \sin(a^{({\rm i})} - b^{({\rm i})}) 
       + \sqrt{3}(-1)^{n^{({\rm i})}}  \cos(a^{({\rm i})} - b^{({\rm i})})  \right]  \lambda_2 
%    \nonumber \\ & \qquad
    \pm \frac{1}{2}  \left\{  (-1)^{n^{({\rm i})}}  
    \left[  3  \sin ( a^{({\rm i})} )   \pm i \cos (a^{({\rm i})} ) \right]
                                     + \sqrt{3}e^{\pm ia^{({\rm i})}}  \right\}   \lambda_4 
    \nonumber \\ & \qquad
    \pm \frac{1}{2}  \left\{  (-1)^{n^{({\rm i})}}  
    \left[ 3  \cos  (a^{({\rm i})} ) \mp i \sin ( a^{({\rm i})} )  \right] 
    \pm i\sqrt{3}e^{\pm ia^{({\rm i})}}  \right\}
    \lambda_5 
    \nonumber \\ & \qquad
    \pm \frac{1}{2}  \left\{  (-1)^{n^{({\rm i})}}
    \left[  3\sin ( b^{({\rm i})} ) \pm  i \cos ( b^{({\rm i})} )  \right] 
    - \sqrt{3}e^{\pm ib^{({\rm i})}}  \right\}
    \lambda_6 
    \nonumber \\ & \qquad
    \pm \frac{1}{2}  \left\{  (-1)^{n^{({\rm i})}}
    \left[  3\cos (b^{({\rm i})}) \mp i\sin (b^{({\rm i})})  \right] 
    \mp i\sqrt{3}  e^{\pm ib^{({\rm i})}}  \right\}
    \lambda_7  
%    \nonumber \\ & \qquad
    \mp \sqrt{3}(-1)^{n^{({\rm i})}}\lambda_3
    + \sqrt{3}i\lambda_8 
    \Big) , \quad \text{and}  \\
    % % %
    L_{\pm 6} 
    & = \frac{ie^{\mp\frac{i}{2}  
    \left(  \phi_3^{({\rm i})} - \phi_1^{({\rm i})}  \right)}}{6}
    \Big(  -i  \left[  \cos  \left(  a^{({\rm i})}- b^{({\rm i})}  \right) 
   + \sqrt{3}(-1)^{n^{({\rm i})}}  \sin(a^{({\rm i})} - b^{({\rm i})})  \right]  \lambda_1  
   \nonumber \\ & \qquad
    + i  \left[  \sin(a^{({\rm i})} - b^{({\rm i})}) 
    - \sqrt{3}(-1)^{n^{({\rm i})}}\cos(a^{({\rm i})} - b^{({\rm i})})  \right]\lambda_2 
%    \nonumber \\ & \qquad
    \mp \frac{1}{2}  \left\{  (-1)^{n^{({\rm i})}}
    \left[  3\sin (a^{({\rm i})}) \pm i  \cos (a^{({\rm i})})  \right] 
    - \sqrt{3}e^{ia^{({\rm i})}}  \right\}
    \lambda_4 
    \nonumber \\ & \qquad
    \mp \frac{1}{2}  \left\{  (-1)^{n^{({\rm i})}}
    \left[  3\cos (a^{({\rm i})} ) \mp i\sin (a^{({\rm i})} )  \right]
    \mp i\sqrt{3}e^{ia^{({\rm i})}}  \right\}
    \lambda_5 
%    \nonumber \\ & 
    \mp \frac{1}{2}  \left\{ (-1)^{n^{({\rm i})}}  
    \left[ 3\sin ( b^{({\rm i})} ) \pm i \cos ( b^{({\rm i})} ) \right] 
    + \sqrt{3}e^{ib^{({\rm i})}} \right\}
    \lambda_6 
    \nonumber \\ & \qquad
    \mp \frac{1}{2}  \left\{  (-1)^{n^{({\rm i})}}
    \left[ 3\cos (b^{({\rm i})} ) \mp i\sin (b^{({\rm i})} ) \right]
    \pm i\sqrt{3}e^{ib^{({\rm i})}}  \right\}
    \lambda_7  
%    \nonumber \\ &
    \mp \sqrt{3}(-1)^{n^{({\rm i})}}\lambda_3
    - \sqrt{3}i\lambda_8 
    \Big) .
 \end{align}
$L_{\pm 7}$, $L_{\pm 8}$, and $L_{\pm 9}$ have the same forms.
However, $({\rm ii})$'s replace the superscripts $({\rm i})$'s. 
$L_{\pm 10}$, $L_{\pm 11}$, and $L_{\pm 12}$ likewise have the same form,
except that $({\rm iii})$'s replace the $({\rm i})$'s.

\end{appendices}

\bibliographystyle{naturemag}
\bibliography{apssamp}
\end{document}